\documentclass[ASNA,twocolumn]{USG} 
\usepackage{anyfontsize} %
\usepackage{mathtools}
\mathtoolsset{showonlyrefs=true}

\usepackage{colortbl}   
\usepackage{xcolor}     

\newcommand{\changed}[1]{{#1}}
\usepackage{makecell}
\usepackage{rotating} 
\usepackage{tikz}
\usetikzlibrary{shapes,arrows,positioning,fit,backgrounds}

\usepackage{steinmetz}
\graphicspath{{./images/}}

\articletype{ORIGINAL ARTICLE}%

\received{}
\revised{}
\accepted{}
\journal{AIS Journal}
\volume{}
\copyyear{2025}
\startpage{}
\articledoi{}

\begin{document}
\thispagestyle{empty}
This manuscript is the peer-reviewed accepted version of an article accepted for publication in Advanced Intelligent Systems.
The final published version will be available via Wiley Online Library at DOI: 10.1002/aisy.202501223.

This article may be used for non-commercial purposes in accordance with Wiley Terms and Conditions for Use of Self-Archived Versions. This article may not be enhanced, enriched or otherwise transformed into a derivative work, without express permission from Wiley or by statutory rights under applicable legislation. Copyright notices must not be removed, obscured or modified. The article must be linked to Wiley’s version of record on Wiley Online Library and any embedding, framing or otherwise making available the article or pages thereof by third parties from platforms, services and websites other than Wiley Online Library must be prohibited.

\newpage

\title{``It Is Much Safer to Be Sparse than \changed{Connected}'': Safe Control of Robotic Swarm Density Dynamics with PDE-Optimization with State Constraints
}
\transtitle{``It Is Much Safer to Be Sparse than Connected'': Safe Control of Robotic Swarm Density Dynamics with PDE-Optimization with State Constraints
}

\author{Longchen Niu}[https://orcid.org/0009-0001-1901-7991]
\author{Gennaro Notomista}[https://orcid.org/0000-0002-1478-2790]

\authormark{NIU \textsc{et al.}}
\titlemark{``It Is Much Safer to Be Sparse than Dense'': Safe Control of Robotic Swarm Density Dynamics with PDE-Optimization with State Constraints}

\address{\orgdiv{Ecological and Resilient Autonomous roBots Lab (ERABLab), Department of Electrical and Computer Engineering, }\orgname{University of Waterloo, }%
\orgaddress{\state{Ontario, }\country{Canada}}}

\corres{Longchen Niu  (\email{l3niu@uwaterloo.ca})}

\editor{\textbf{Academic Editor:} XXX ~|~ \textbf{Guest Editor:} XXX}

\presentaddress{Univeristy of Waterloo}

\fundingInfo{This work has been partially supported by the NSERC Discovery Grant RGPIN-2023-03703.}

\keywords{PDE Optimization-based Control | Swarm Density Control | Safety-critical Applications | Voronoi-based Methods}

\transkeywords{PDE Optimization-based Control | Swarm Density Control | Safety-critical Applications | Voronoi-based Methods}

\abstract[ABSTRACT]{
This paper introduces a safety-critical optimization-based control strategy that leverages control Lyapunov and control barrier functions to guide the spatial density of robotic swarms governed by the Fokker–Planck equation to a predefined target distribution. In contrast to traditional open-loop state-constrained optimal control strategies, the proposed approach operates in closed-loop, and a Voronoi-based variant further enables distributed deployments. Theoretical guarantees of safety are derived, and numerical simulations demonstrate the performance of the proposed controllers. Finally, a multi-robot experiment showcases the real-world applicability of the proposed controllers under localization and motion noises, illustrating how it is much easier for a sparse swarm to satisfy safety specifications than it is for a densely packed one.
}

\abbr{PDE, partial differential equation; ODE, ordinary differential equation; FDM, finite difference method; FEM, finite element method.}

\contributed{Longchen Niu conducted this study under the supervision and guidance of Gennaro Notomista}


\copyright{This is an open access article under the terms of the \href{Creative Commons Attribution-NonCommercial}{Creative Commons Attribution-NonCommercial} License, which permits use, distribution and reproduction in any medium, provided the
original~work~is~properly cited and is not used for commercial purposes.
\\[5pt]
  ©  2025 The Author(s) \textit{AIChE Journal} published by Wiley Periodicals LLC on behalf of American Institute of Chemical Engineers.}


\maketitle


\section{Introduction}\label{Sec1: intro}

Density control in swarm robotics is essential for numerous applications, including environmental monitoring, wildfire suppression, search and rescue, agriculture, and medical diagnostics \cite{application_enivronment, Fire, application_SaR, application_agriculture, application_med}. By guiding the spatial distribution of a robotic swarm to align with a target Probability Density Function (PDF), these systems can autonomously perform critical tasks. Examples include mapping environmental parameters such as humidity or sunlight, deploying fire suppression agents in inaccessible locations, locating survivors in complex post-disaster terrains, optimizing crop management by identifying areas with low moisture or insufficient sunlight, and targeting medical interventions by concentrating on infected regions within the body \cite{8206299}. In many of these applications, it is paramount to be able to enforce invariance constraints, consisting of keeping the state of the system confined within a subset of the state space. This constraint is crucial, for instance, to protect both robotic assets and the surrounding environment, enhancing safety and operational robustness in real-world deployments. Approaches to ensure safety constraints in the context of multi-robot systems modeled by Ordinary Differential Equations (ODEs) are proposed, for instance, in \cite{wang2017safety, CBF, safety_car_distance}.

Optimal control approaches have traditionally been employed to ensure the optimality of the system's cost function over a, possibly infinite, time horizon. Recently developed techniques allow us to encompass state invariance constraints within optimal control frameworks (e.g. \cite{OptimalMain, exisitenceMeanFieldGameWithSTate,quasilinear_parabolic_PDE_OptimalControl_StateConstraint}). For ODEs, optimization-based control (OBC) frameworks are becoming increasingly popular thanks to the possibility of solving optimization problems in real-time in the tight control loop. Within this category, control barrier functions (CBFs) have emerged as a powerful technique for maintaining invariance \cite{CBF}. However, the extension of these safety or invariance methods to systems modeled by partial differential equations (PDEs) remains comparatively underexplored.

In this paper, we propose a novel density control strategy for robotic swarms that integrates CBFs with PDE-based formulations to enforce safety constraints. Our approach models the spatial density of teams of robots as PDFs to leverage the Fokker-Planck equation. This allows us to capture uncertainties from sensor and motion noise, thus providing robustness and scalability suitable for practical deployment as seen in \cite{RalPaper}. Building on this foundation, we develop an OBC that enforces swarm safety guarantees and introduce a computationally efficient variant based on Voronoi partitions. This speed-up not only reduces the typical scaling issue associated with centralized control, but also challenges the common stereotypes that a swarm is more efficient with more agents working in a confined space.
We further derive theoretical guarantees of state invariance for swarm safety and show through simulations that our OBC maintains safety in the presence of localization and motion noise, when the state-of-the-art \changed{optimal control (OC)} fails. Finally, we validated the accelerated OBC in a six-robot experiment, confirming its practical applicability. 

\section{Related Work}\label{sec2: related work}
Traditional OC methods have recently been explored for swarm density control. For instance, the authors of \cite{Optimal_noInvariance} employed an optimal control formulation based on a more general advection-diffusion-reaction PDE to achieve the desired swarm density distributions. However, this approach does not include state invariance, allowing nonzero densities in certain task regions during transient states, potentially compromising robot safety. In contrast, the novel formulation we propose in this paper explicitly maintains spatial state invariance throughout the control horizon, thereby ensuring consistent swarm safety conditions.

OC with set invariance has also been examined in recent studies. In \cite{Optimal_withInvariance}, for instance, invariance is enforced by explicitly excluding designated regions from the computational domain, solving instead for optimal control fields that stabilize swarm density outside these areas. Although computationally efficient and straightforward to implement, this exclusion-based approach is limited to static, pre-defined environments where the invariance constraint enforces zero density in a specified area at all times. Conversely, the OC method we examine is proposed in \cite{OptimalMain} and directly imposes state constraints over the entire domain. This allows the incorporation of time-varying constraints, such as dynamic hazard regions with moving obstacles or fluctuating wind conditions for drones. Additionally, it naturally accommodates time-varying target densities, facilitating a wider range of applications.

As regards to OBC, control Lyapunov functions (CLFs) have been applied to PDE-governed multi-agent systems. For example, \cite{CLF_PDE_NoDensity} utilized CLFs to achieve formation control in higher-order PDE-based systems. However, this approach relies on precise position measurements, making it vulnerable to measurement inaccuracies. To address this limitation, our formulation explicitly includes both measurement and motion noise within the Fokker-Planck equation, enhancing robustness and resilience to imperfect position measurements and motion actuation in practical applications.

Within the OBC framework, CBFs have been extensively studied to ensure robot safety in ODE-governed systems, as demonstrated in several examples mentioned in \cite{CBF}. Recent extensions of CBFs to PDE-governed systems employed Finite Element Methods (FEM) for discretization, rigorously addressing discretization errors to achieve high precision in the critical application of flexible structure control in fusion reactors \cite{CBF-PDE2ODE}. In contrast, our method uses the Finite Difference (FD) discretization of the Fokker-Planck equation for swarm density control. Although FD typically exhibits larger discretization errors compared to FEM, these errors remain acceptable given the smooth density fields on a uniform grid in our application, instead of irregular geometries in \cite{CBF-PDE2ODE}. However, for applications involving robots with discontinuous PDFs, FEM or Finite Volume Methods might be necessary for numerical stability. Consequently, the optimization-based feedback control we propose in this letter exhibits significantly lower computational complexity with the FD method, allowing even for an in-the-loop calculation of the control inputs.

\changed{Our previous work, \cite{RalPaper}, has shown} through mathematical proofs and experiments that incorporating measurement and motion noise through PDF and Fokker-Planck equation enhances robustness in practical applications. Building on this foundation, the work in \cite{MRS_paper} introduced a decentralized CLF-CBF OBC framework to enforce local safety under noise, provided formal mathematical bounds to quantify the difference between local and global safety, and validated the approach through quadcopter experiments. In parallel, this current paper focuses on centralized approaches to provide guarantees of global safety under noise. In addition, we compare our method with the state-of-the-art OC and introduce a Voronoi-based variant that improves scalability and performance, with validation through both simulations and experiments. 

\subsection{Contributions}
The main contributions of this work are as follows: 
\begin{itemize}
    \item \changed{We extend PDE-based swarm density control} with noise by incorporating CLF-CBF constraints, and provide rigorous safety guarantees even in the presence of localization and motion uncertainties. 
    \item We develop a Voronoi-based controller variant that significantly reduces the computational scaling challenge of the centralized controller, enabling scalable and real-time deployments. 
    \item We show through both theoretical proof and simulations that, contrary to the common stereotypes of swarms being more effective in large numbers and confined spaces, safety specifications are most effectively satisfied in sparse teams.
    \item We demonstrate through simulations and a six-robot experiment that our controller follows a target PDF while maintaining swarm safety even with noise, where the state-of-the-art optimal control fails. 
\end{itemize}

\section{Mathematical Formulation}\label{sec3: math formulation}
\begin{figure*}
    \centering
    \includegraphics[width=1.0\linewidth]{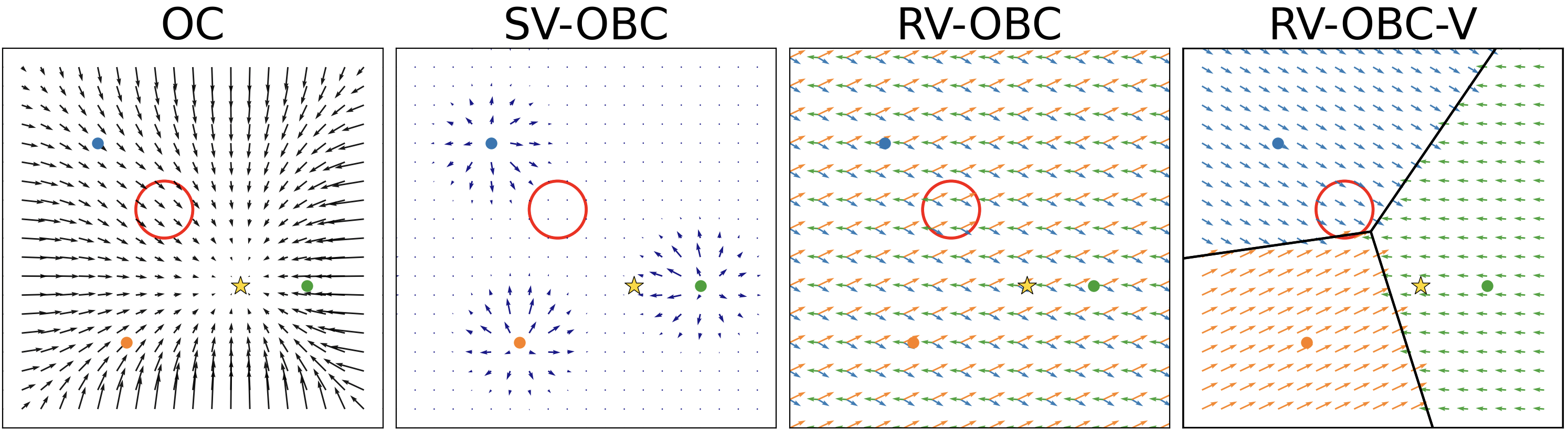}
    \caption{Comparison of four control strategies with \(3\)~robots (blue, orange, and green dots) \changed{with a red danger zone and a gold star target}. From left to right: OC precomputes a global control field; SV-OBC generates a new control field per step, focusing on where the robot PDF is high; RV-OBC produces a separate uniform field per robot; RV-OBC-V partitions the spatial domain and optimizes uniform control fields within each cell.}
    \label{fig: 4plot}
\end{figure*}
The primary objective of this work is to design a density-based controller for a team of robots whose motion is governed by stochastic dynamics. Specifically, the robots aim to match a predefined target PDF autonomously, representing a desired spatial distribution. To accurately capture the motion uncertainties inherent in swarm robotic systems, we employ a continuous-time Itô stochastic differential equation as presented in \cite{ANNUNZIATO2013487}:
\begin{equation}
    dX_t = b(X_t, t; u)\,dt + \sigma(X_t, t)\, dW_t,\quad X_{t_0}=X_0,
    \label{eq: general_model}
\end{equation}
where the state \(X_t \in \mathbb{R}^2\) denotes the robot position, the drift term \(b(X_t, t; u)\) models deterministic control input, and \(\sigma(X_t, t)\) represents stochastic perturbations proportional to a Wiener process \(W_t\in\mathbb{R}^m\) with independent increments.

While equation \eqref{eq: general_model} describes the noisy motion of individual robots \cite{doi:10.1177/0278364908100177}, swarm density control requires a macroscopic representation of the spatial robot density. Leveraging the averaging technique from \cite{ANNUNZIATO2013487, Archer2004DynamicalDF} and employing Itô calculus, we derive the Fokker-Planck PDE for the robot density dynamics with \changed{the per-robot control fields \(u_{\mathrm{fields(i)}}(r,t): \mathbb{R}^2 \times \mathbb{R} \;\to\; \mathbb{R}^2 \):
\begin{equation}
    \frac{\partial \rho(r, t)}{\partial t} = \sum_{i \in N} \bigg[-\nabla_i(u_{\mathrm{fields(i)}}(r,t) \  \rho_i(r,t)) + T\Delta_i \rho_i(r, t) \bigg],
    \label{eq:FP_equation}
\end{equation}}
where the robot ID is the subscript \(i\), the diffusion constant \(T\) represents the maximum magnitude of uniform Gaussian motion noise \cite{RalPaper}. 
\changed{
Here, \(\rho(r,t)\) should be interpreted as a belief-weighted physical probability mass of the swarm. Each \(\rho_i(r,t)\) represents a probabilistic belief (up to normalization) of robot \(i\)'s true position with the inverse covariance (precision) matrix \(\Sigma(i)\), encoding the confidence in the measured position \(x_i\). At the extreme, when measurement is perfect, \(\Sigma(i) \rightarrow 0\), \(\rho_i(r,t)\) becomes a delta function at the robot's true position, and the proposed PDE-based framework is no longer applicable. Furthermore, we assume identical localization uncertainty for all robots, \(\Sigma = \Sigma(i) \forall i \in N\), to reflect identical sensing hardware.
}
Therefore, the swarm density \(\rho(r, t)\) at position \(r\) and time \(t\) can be represented by a sum over \(N\) individual robot positions \(x_i\):
\begin{equation}
    \rho(r, t) = \sum_{i=1}^{N} \rho_i = \sum_{i=1}^{N}  \exp\left(-\frac{1}{2}(r - x_i)^T \Sigma (r - x_i)\right).
\label{eq: rho def}
\end{equation}

\changed{
The Gaussian kernels in \eqref{eq: rho def} can be normalized by a constant factor depending only on \(\Sigma\); this constant scaling can be absorbed into the density and associated thresholds without affecting the subsequent analysis. 
Under this interpretation, the Fokker-Planck equation \eqref{eq:FP_equation} describes the evolution of the collective belief distribution over the robots' position measurements.} 
Importantly, equation \eqref{eq:FP_equation} is deterministic, as it characterizes the swarm density distribution from a statistical perspective rather than tracking individual stochastic trajectories in \eqref{eq: general_model}.

Having established the background knowledge on how the noise can be naturally encoded in PDF and Fokker-Planck equations, we step back to take a high-level look at the four controllers proposed in the remainder of this section. Figure~\ref{fig: 4plot} shows an example with \(3\)~robots highlighting their distinct control strategies.

\subsection{Optimal Control}
We define the control objective as the cost functional:\changed{
\begin{equation}
\begin{aligned}
J(\rho, u_\mathrm{field}) = &\int_0^{t_f} \int_\Omega \frac{\alpha}{2}\|u_\mathrm{field}(r,t)\|^2 \rho(r,t) \ dr \ dt \\
+ &\int_0^{t_f} \frac{1}{2}\|\rho_d(r,t) - \rho(r,t)\|_{L^2(\Omega)}^2 \ dt \\
+ &\frac{1}{2}\|\rho_d(r,t_f)-\rho(r,t_f)\|_{L^2(\Omega)}^2,
\end{aligned}
\label{eq:cost_function}
\end{equation}}
where \({t_f}\) is the terminal time, \(\rho_d(r,t)\) denotes the desired target density, \(\Omega \subseteq \mathbb{R}^2\) is a 2D spatial domain, \(u_\mathrm{field} \) is a single control input field applied identically to all robots, i.e., \changed{\(u_\mathrm{field}= u_\mathrm{fields(i)} \, \forall i \in N\)}, and \(\alpha\) is a weight coefficient balancing density matching accuracy and control effort. This formulation aligns with the optimal control problem described in \cite{OptimalMain}, corresponding to the terms: 
\changed{
\begin{equation}
\begin{aligned}
    &L=\frac{\alpha}{2}\|u_\mathrm{field}(r,t)\|^2, \quad F=\frac{1}{2}\|\rho_d(r,t)-\rho(r,t)\|_{L^2(\Omega)}^2,\\
    &G=\frac{1}{2}\|\rho_d(r,t_f)-\rho(r,t_f)\|_{L^2(\Omega)}^2.
\end{aligned}
\label{eq:Representaitons}
\end{equation}
}

To employ the optimal control formulation in \cite{OptimalMain}, we utilize the Hamiltonian \(H\), defined as the convex conjugate of \(L\) (see \cite{briani2016stablesolutionspotentialmean}), \(H(r, p) = \sup_{q}\{-p\cdot q - L(r, q)\} = \frac{1}{2\alpha}\|p\|^2\), and \(\nabla_p H(r, p)=\frac{1}{\alpha}p\).

In practical scenarios, spatial safety constraints can be mathematically encoded using a state invariance constraint:
\begin{equation}
    \psi(\rho) = \int_{\mathcal{A}}\rho^2(r,t)\,dr - \epsilon \leq 0,
\end{equation}
where \(\mathcal{A}\subseteq\Omega\) represents the restricted area and \(\epsilon\) is \changed{the specified threshold on the aggregated belief mass within \(\mathcal{A}\), representing an acceptable level of residual uncertainty due to localization noise.}

With these considerations, the corresponding optimality conditions, incorporating the safety constraint as derived from Theorem 2.2 in \cite{OptimalMain}, can be summarized as:
\begin{equation}
    u_\mathrm{field}=-\nabla_p H(r,\nabla w)=-\frac{1}{\alpha}\nabla w,
\end{equation}
with \(w = w(r,t)\) being the adjoint state (solution of the Hamilton-Jacobi-Bellman equation in \cite{OptimalMain}), subject to the coupled PDE system:
\begin{equation}
\begin{aligned}
\begin{cases}
    &-\partial_t w(r,t) + \frac{1}{2\alpha}\|\nabla w(r,t)\|^2 - T\Delta w(r,t) = \\& \hspace{2.2cm}2\nu(t)\int_{\mathcal{A}}\rho(r,t)\,dr - \int_{\Omega}(\rho_d-\rho)\ dr,\\
    &\partial_t\rho(r,t) + \nabla(u_\mathrm{field}\ \rho(r,t)) - T\Delta\rho(r,t)=0,\\
    &w(r,{t_f})=2\eta\int_{\mathcal{A}}\rho(r,{t_f})\,dr - \int_{\Omega}(\rho_d({t_f})-\rho({t_f})) \ dr,\\
    &\rho(r,0)=\rho_0(r),
\end{cases}
\end{aligned}
\label{eq:Optimal_System}
\end{equation}
where the Lagrange multipliers \(\nu(t)\geq0\) and \(\eta\geq0\) enforce the safety constraint, being strictly positive when the constraint is active, i.e., when \(\psi(\rho)=0\), and \(0\) otherwise. 

\subsection{OC \textsc{vs} OBC Frameworks}
\label{sec: Opt vs Opt}
A fundamental difference between the two types of controllers, OC and OBC, lies in their approach to feedback: open-loop versus closed-loop control. OC exemplifies an open-loop strategy, where control inputs are predetermined and independent of real-time system outputs. Specifically, in \eqref{eq:Optimal_System}, a control field is computed solely with the initial density distribution of the robots. Then, the robots navigate according to this precomputed field, as illustrated in Fig.~\ref{fig:BlockDiagram Optimal}. The dashed line in the diagram denotes a single communication event at the initialization phase, after which the system evolves without further input from the controller. Although OC formally uses state feedback, this information is not used for optimization but just for a static lookup table. Consequently, the controller does not adapt to any disturbances, an ability typically implied with feedback control, and can quickly lose optimality. 
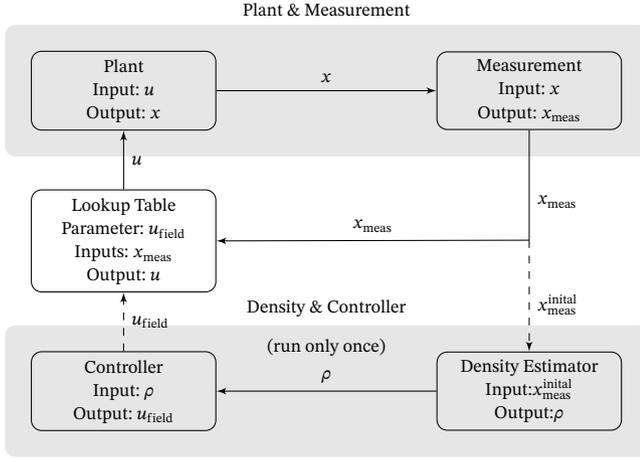
\begin{figure}[t]
    \centering
    \resizebox{\linewidth}{!}{\usetikzlibrary{shapes,arrows,positioning,fit,backgrounds}

\begin{tikzpicture}[auto, >=latex', node distance=2cm, every node/.style={font=\footnotesize}]
    \tikzstyle{block} = [draw, rectangle, rounded corners,
        minimum height=2.5em, minimum width=8em, align=center];
    \tikzstyle{line} = [draw, -latex'];
    \tikzstyle{group} = [fill=gray!20, rounded corners, inner sep=10pt, draw=none];

    \node[block] (plant) 
        {Plant\\Input: \(u\)\\Output: \(x\)};
    \node[block, right=3cm of plant] (measure) 
        {Measurement\\Input: \(x\)\\Output: \(x_{\mathrm{meas}}\)};

    \node[block, below=0.8cm of plant] (compute) 
        {Lookup Table\\Parameter: \(u_{\mathrm{field}}\)\\ Inputs: \(x_{\mathrm{meas}}\)\\Output: \(u\)};

    \node[block, below=0.8cm of compute] (controller) 
        {Controller\\Input: \(\rho\)\\Output: \(u_{\mathrm{field}}\)};
    \node[block, right=3cm of controller] (density) 
        {Density Estimator\\Input:\(x_{\mathrm{meas}}^\mathrm{inital}\)\\Output:\(\rho\)};

    \begin{pgfonlayer}{background}
        \node[group, fit=(plant)(measure), 
              label=above:{\footnotesize Plant \& Measurement}] (topgroup) {};
        \node[group, fit=(density)(controller), 
          label=above:{\footnotesize Density \& Controller}] (botgroup) {};
    \node at ([yshift=-0.3cm]botgroup.north) { (run only once)};
\end{pgfonlayer}
    \draw[line] (plant) -- node[pos=0.5, above] {\(x\)} (measure);

    \draw[-] (measure.south) -- ++(0,-1.5) 
    node[right, yshift=15pt]{\(x_{\mathrm{meas}}\)}
    coordinate (splitpt);

    \draw[line] (splitpt) -- (compute.east) node[midway,above] {\(x_{\mathrm{meas}}\)};
    
    \draw[line, dashed, -latex'] (splitpt) -| node[midway, yshift=-25pt] {\(x_{\mathrm{meas}}^\mathrm{inital}\)} (density.north);

    \draw[line] (density) -- node[above] {\(\rho\)} (controller);

    \draw[line, dashed, -latex'] (controller.north) -- node[right] {\(u_{\mathrm{field}}\)} (compute.south);

    \draw[line] (compute.north) -- node[right] {\(u\)} (plant.south);

\end{tikzpicture}}
    \caption{System diagram for OC. The control field is computed once at initialization and remains unchanged throughout the operation.}
    \label{fig:BlockDiagram Optimal}
\end{figure}

In contrast, OBC operates within a closed-loop framework, continuously adjusting the control inputs based on real-time feedback. At each time step, the system measures the robots' positions and uses this data to compute an updated control command for the next time step. This feedback-driven approach, depicted in Fig.~\ref{fig:BlockDiagram Optimization}, allows the controller to dynamically respond to any disturbances and uncertainties in the system.

While the OC method possesses a theoretical guarantee for convergence, the numerical implementation of OC presents significant challenges. For instance, the state invariance constraint is inherently discontinuous. As the robot approaches the hazard area, the Lagrange multipliers \(\nu(t)\) and \(\eta\) in \eqref{eq:Optimal_System} jump from zero to a finite value, leading to a sudden rise in an area where the density of the whole field was previously smooth, causing discontinuities along the boundary \(\partial \mathcal{A}\). 
Compared to OC, an OBC exhibits significantly less numerical instability. This is primarily because it solves for only a single timestep at a time with a smooth resetting density field. In contrast, OC must account for the entire temporal field, allowing numerical errors to accumulate and potentially destabilize the system. The remainder of this section introduces three distinct OBCs, beginning with the OC-inspired SV-OBC.

\subsection{Swarm Velocity Optimization-Based Control (SV-OBC)}
Due to its similarity with the well-established CLFs, CBFs can be naturally combined with CLF methods in an optimization-based framework, which has shown success in a wide range of real-time robotic applications \cite{wang2017safety, CBF}. Instead of directly optimizing the cost defined in \eqref{eq:cost_function}, the approach we propose in this paper starts by formulating the control objective in terms of a CLF constraint, thereby providing explicit guarantees for stability and convergence. Specifically, we define a CLF as a measure of deviation from the desired density, \(F\) in \eqref{eq:Representaitons}, and reformulate the cost function \eqref{eq:cost_function} into an optimization problem where the primary objective is the minimization of the control field \( u_\mathrm{field} \). \changed{Notably, since \(F(t_f) = G(t_f)\) in \eqref{eq:Representaitons}, omitting \(G\) in the optimization framework is acceptable for two reasons. First, as \(t_f\) increases, the contribution from \(G\) diminishes compared to the running cost \(F\).} Second, by using \(F\) as the Lyapunov function, this formulation inherently minimizes \(F\) and, consequently, the terminal cost \(G\) whenever possible.

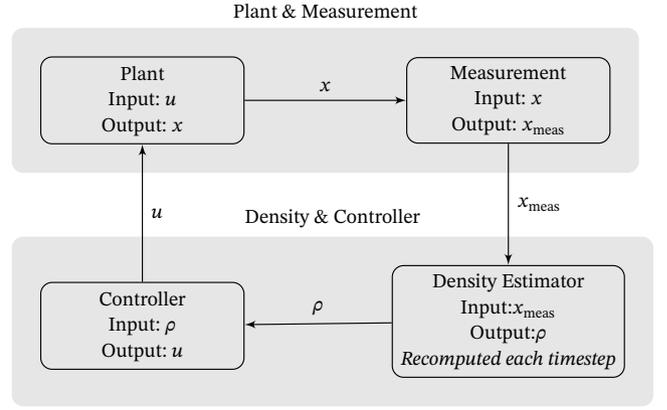
\begin{figure}[t]
    \centering
    \resizebox{\linewidth}{!}{\usetikzlibrary{shapes,arrows,positioning,fit,backgrounds}

\begin{tikzpicture}[  auto, >=latex',
  node distance=2cm,
  every node/.style={font=\footnotesize},
  every label/.style={font=\footnotesize}]
    \tikzstyle{block} = [draw, rectangle, rounded corners,
        minimum height=2.5em, minimum width=8em, align=center];
    \tikzstyle{line} = [draw, -latex'];
    \tikzstyle{group} = [fill=gray!20, rounded corners, inner sep=10pt, draw=none];

    \node[block] (plant) 
        {Plant\\Input: \(u\)\\Output: \(x\)};
    \node[block, right=2cm of plant] (measure) 
        {Measurement\\Input: \(x\)\\Output: \(x_{\mathrm{meas}}\)};

    \node[block, below=1.7cm of plant] (controller) 
        {Controller\\Input: \(\rho\)\\Output: \(u\)};
    \node[block, below=1.5cm of measure] (density) 
    {Density Estimator\\Input:\(x_{\mathrm{meas}}\)\\Output:\(\rho\)\\\emph{Recomputed each timestep}};

    \begin{pgfonlayer}{background}
        \node[group, fit=(plant)(measure), 
              label=above:{\footnotesize Plant \& Measurement}] (topgroup) {};
        \node[group, fit=(density)(controller), 
              label=above:{\footnotesize Density \& Controller}] (botgroup) {};
    \end{pgfonlayer}

    \draw[line] (plant) -- node[pos=0.5, above] {\(x\)} (measure);

    \draw[line] (measure.south) -- node[right] {\(x_{\mathrm{meas}}\)} (density);
    \draw[line] (density) -- node[above] {\(\rho\)} (controller);

    \draw[line] (controller.north) -- node[right] {\(u\)} (plant.south);

\end{tikzpicture}}
    \caption{System diagram for OBC. The entire loop is recomputed at each timestep to achieve real-time feedback.}
    \label{fig:BlockDiagram Optimization}
\end{figure}

The safety constraint, typically enforced via a positive barrier function \(h\), is defined as the additive inverse of the original constraint function \(\psi\): \(h(\rho) = -\psi(\rho) = \epsilon - \int_{\mathcal{A}}\rho^2(r)\,dr\). From this point onward, we omit the explicit time dependence in the notation. Since the optimization problem is solved at the current timestep, all quantities are implicitly evaluated at the same timestep \(t_k\). 
With this definition, the OBC incorporating both CLF and CBF constraints can be formulated as:
\begin{equation}
    \begin{aligned}
    \min_{u_\mathrm{field},\,s} &\quad \|u_\mathrm{field}\|_{L^2(\Omega)}^2 + \gamma s\\
    \quad \text{s.t.} &\quad \alpha_v V(\rho) + \dot{V}(\rho,u_\mathrm{field}) - s\leq 0\\
    &\quad \alpha_h h(\rho) + \dot{h}(\rho,u_\mathrm{field}) \geq 0,
\end{aligned}
\label{eq:OptimizationBased}
\end{equation}
where \(V(\rho) = F = \frac{1}{2}\|\rho_d - \rho\|_{L^2(\Omega)}^2\) is the positive semi-definite Lyapunov function that encourages convergence toward the desired density distribution. Here, \(\alpha_v, \alpha_h > 0\) determine the convergence rate toward the CLF and CBF constraints, respectively. 
\changed{The positive weighting factor \(\gamma\) penalizes deviation from the CLF objective through the slack variable \(s \geq 0\). The CBF is imposed as a hard safety constraint, while the CLF is softened via \(s\) to ensure the feasibility of the solution.}
The time derivatives, \(\dot{V}\) and \(\dot{h}\), are given by:
\begin{equation}
\begin{aligned}
    \dot{V}(\rho,u_\mathrm{field}) = - \int_{\Omega} (\rho_d - \rho) \rho_t\,dr, \quad
    \dot{h}(\rho,u_\mathrm{field}) = -\int_{\mathcal{A}} 2\rho \ \rho_t\,dr,
\end{aligned}
\end{equation}
where \(\rho_t(u_\mathrm{field},r)\) is a one-step Fokker-Planck equation based on \eqref{eq:FP_equation}.
This systematic framework easily accommodates varying numbers of robots, as the density field is constructed with any robot quantity without computational overhead.

To implement the infinite-dimensional \changed{optimization problem} \eqref{eq:OptimizationBased}, we employ a central FD scheme with periodic boundary conditions \changed{(see Remark~\ref{remark: boundary conditions})}. This means, on a uniformed 2D grid of size \(N_x \times N_y\) with spacing \(l\), the \changed{continuous density fields} \(\rho, \rho_d \in L^2(\Omega)\) and \changed{control field} \(u_\mathrm{field} \in L^2(\Omega; \mathbb{R}^2)\) are approximated and then flatten as \(\hat{\rho},\hat{\rho}_d\in \mathbb{R}^{N_d}\) and \(\hat{u}\in \mathbb{R}^{2N_d}\), respectively at grid points with \(N_d = N_x N_y\). Then, \eqref{eq:OptimizationBased} becomes:
\begin{equation}
    \begin{aligned}
    \min_{\hat{u},s} &\quad \hat{u}^T \, \hat{u} l^2 + \gamma s\\
    \quad \text{s.t.} &\sum_{i=1}^{N_d} \left[\frac{\alpha_v}{2} \left(\hat{\rho}_d^{i} - \hat{\rho}^{i} \right)^2 - \left( \hat{\rho}_d^{i} - \hat{\rho}^{i} \right) \hat{\rho}_t^{i} \right]l^2 \leq s\\
    &\alpha_h \epsilon + \sum_{i \in \mathcal{A}} \Big[ - \alpha_h (\hat{\rho}^{i})^2  - 2  \hat{\rho}^{i} \hat{\rho}_t^{i} \Big]l^2  \geq 0\\
    & \hat{u}^{i} \in [\underline{u}, \overline{u}]\quad\forall i,
\end{aligned}
\label{eq:OptimizationBased_discretized}
\end{equation}
where the superscript \(i\) denote the flattened grid position. \(\hat{\rho}_t\) is the ODE approximation of the PDE \eqref{eq:FP_equation} using the FD method: \(\hat{\rho}_t =  A(\hat{\rho}) \hat{u} + T \, B \hat{\rho}, \)
\changed{
where \(A(\hat\rho)\) and \(B\) are constant sparse matrices obtained from central FD, with the five-point Laplacian \(B\) multiplied by the scalar diffusion coefficient \(T\).}
The last constraint represents a component-wise bound on the control input with:
\begin{equation}
    [\underline{u}, \overline{u}] \supseteq \left[ \min \left( -A^{\dagger}(\hat{\rho})\, T B \hat{\rho} \right),\ \max \left( -A^{\dagger}(\hat{\rho})\, T B \hat{\rho} \right) \right],
\end{equation}
where \(A^\dagger\) is the Moore–Penrose pseudoinverse of \(A\). 
With this, it is clear to see that system \eqref{eq:OptimizationBased_discretized} is a convex quadratic cost optimization problem with affine constraints, enabling fast computational speed for real-time implementation.

\begin{lemma}
\label{lem:pseudoinverse}
Using a central-difference scheme with periodic or homogeneous Neumann boundary conditions and odd grid sizes \(N_x, N_y\), \changed{
the discretized diffusion term \(b=-T\,B\hat\rho\), where \(B\) is the standard 5-point Laplacian, lies in the image of the advection operator
\[
A(\hat\rho)=\big[\,D_x\mathrm{Diag}(\hat\rho)\;\; D_y\mathrm{Diag}(\hat\rho)\,\big]\in\mathbb{R}^{N_d\times 2N_d},\qquad N_d=N_xN_y.
\]

Consequently, the linear system \(A(\hat \rho) \hat u = b\) is consistent, and the pseudoinverse solution \(\hat u^*=A^\dagger(\hat\rho)\,b\) yields an exact solution.
}
\end{lemma}

\begin{proof}
Let \(R=\mathrm{Diag}(\hat\rho)\) with \(\hat\rho_i>0\). Under periodic or homogeneous Neumann BCs, the central-difference matrices satisfy
\(
(D_x)^\top=-D_x,\ (D_y)^\top=-D_y.
\) Thus
\[
A^\top z=\begin{bmatrix}R(D_x)^\top z\\ R(D_y)^\top z\end{bmatrix}=0
\quad\Longleftrightarrow\quad
D_x z=0\ \text{and}\ D_y z=0,
\]
since \(R\) is invertible (positive and diagonal).

In 1D, \(D z=0\) means \(z_{i+1}=z_{i-1}\). With \(N\) odd, the map \(i\mapsto i+2\) cycles through all indices, so \(z\) has to be a constant. Applying this in both \(x\) and \(y\) directions yields \(z=c\,\mathbf 1\). Hence \(\ker A^\top=\mathrm{span}\{\mathbf 1\}\) and \(\mathrm{Im}\,A=(\ker A^\top)^\perp=\mathbf 1^\perp\).

For the 5-point Laplacian \(B\) with periodic or homogeneous Neumann BCs, each row sums to zero, so \(\mathbf 1^\top B=0\). Therefore \(b=-T\,B\hat\rho\) satisfies \(\mathbf 1^\top b=0\), i.e., \(b\in\mathbf 1^\perp=\mathrm{Im}\,A\). The system \(A(\hat \rho)\hat u=b\) is thus consistent but underdetermined.
\changed{The Moore–Penrose solution \(\hat u^*=A^\dagger (\hat \rho) b\) then yields the unique solution minimizing \(\| \hat u \|_2\).}
\end{proof}

\begin{remark}
If either $N_x$ or $N_y$ is even, central differences \(i\mapsto i+2\) cycles are broken, causing \((\ker A^\top)^\perp \neq \mathbf 1^\perp\). In this case, a different operator, such as the Finite Volume Method, can be used to find the same solvability result regardless of \(N_x, Ny\). However, since the controller is designed for lightweight online applications, the simplest fix is adjusting the grid size, rather than implementing a more complex scheme. 
\end{remark}

\changed{
\begin{remark} \label{remark: boundary conditions}
    Periodic and Neumann boundary conditions are used in Lemma~\ref{lem:pseudoinverse} for mass preservation to ensure well-posedness of the discretized density dynamics. Periodic boundary conditions are implemented in this work to simplify the FD stencil and avoid boundary-specific complications. They do not imply that robots physically wrap around the workspace. When an explicit workspace boundary exists, Neumann boundary conditions provide a physically meaningful reflective model.

    In practice, the target and invariance sets are chosen strictly inside the computational domain, which can always be inflated without altering the control formulation. As a result, the swarm density mass rarely interacts with the domain boundary during operation, and the closed-loop behavior is identical between periodic and Neumann boundary conditions. 

    Dirichlet boundary conditions, on the other hand, conflict with the swarm density mass conservation and are therefore not considered. If it must be imposed, a finite-volume discretization may be employed at the cost of increased computational complexity.  
\end{remark}
}

\begin{proposition}\label{prop1: control field feasible}
    The constrained optimization problem \eqref{eq:OptimizationBased_discretized} is always feasible. Its unique solution \(\hat{u}^*\) guarantees the forward invariance of the set $\{\hat \rho : h(\hat \rho)\ge0\}$.
\end{proposition}
\begin{proof}
\changed{For any \(h(\hat{\rho})\geq 0\), Lemma~\ref{lem:pseudoinverse} provides an exact solution that yields \(\hat\rho_t=0\), where the controlled advection term cancels out diffusion exactly.
Consequently, \(\dot{h} = 0\) and the CBF constraint is satisfied, ensuring forward invariance of the set $\{\hat \rho: h(\hat \rho)\ge0\}$.} Moreover, the CLF constraint (first constraint in \eqref{eq:OptimizationBased_discretized}) is relaxed by \(s\) and is therefore always feasible. Finally, as the optimization cost is strictly convex, the optimal solution \(\hat{u}^*\) is unique.
\end{proof}

\begin{remark}
    \changed{Since the safety CBF is enforced as a hard constraint in \eqref{eq:OptimizationBased_discretized}, the CLF objective may be relaxed to maintain invariance, potentially leading $\hat\rho$ to converge to a local minimum.} This is chosen as the solution to safety-critical applications. Thanks to the discretization proposed in \eqref{eq:OptimizationBased_discretized}, such local minima can be characterized using the techniques developed in \cite{reis2020control} for ODEs.
\end{remark}

\begin{figure}
    \centering
    \includegraphics[width=0.65\linewidth]{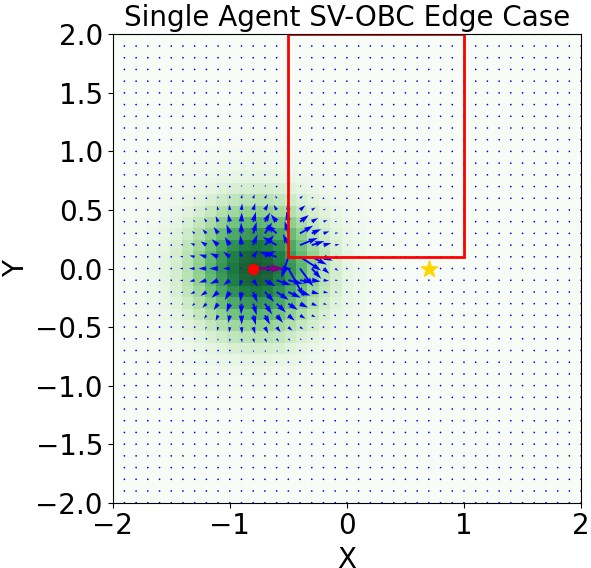}
    \caption{SV-OBC Edge Case. The red area represents the unsafe area, the yellow star is the target center, and the red dot marks the robot’s position with the green PDF. The control field is shown in blue, with purple as the interpolated command for robot.}
    \label{fig: field control edgecase}
\end{figure}

While SV-OBC \eqref{eq:OptimizationBased_discretized} yields a feasible solution that ensures the safety of the swarm density \(\hat{\rho}\) by Proposition~\ref{prop1: control field feasible}, an \changed{important} edge case exists. As shown in Fig.~\ref{fig: field control edgecase}, sufficiently large localization errors can generate conflicting control \changed{values within a robot's belief density. Since the PDF represents localization uncertainty, and the robots cannot execute the full spatial control field, a single control command is interpolated at their most probable location: the center of the PDF. As a result, the control values at the edge of the PDF may not be realized in execution, potentially leading to robot-level unsafe motion despite swarm-level invariance. This behavior does not contradict the swarm-level safety guarantee, but instead reveals a gap between density-level invariance and robot-level execution under localization uncertainty.}

Attempts to minimize this \changed{mismatch} have been made, with additional constraints on the covariance matrix of PDFs and enforcing a \changed{locally} uniform control field within a robot's PDF. Neither of these attempts \changed{proved} satisfactory, even without the additional time complexity; the covariance matrix does not maintain the Gaussian \changed{structure of the belief density}, and the uniform control field fails when multiple PDFs overlap.

\subsection{Robot Velocity Optimization-Based Control (RV-OBC)}
\label{subsec: sum of field}
To avoid the aforementioned drawback, we reformulate our control to \changed{avoid} the OC idea of a single \changed{global} control field applied to all robots, \(u_\mathrm{field}\) in \eqref{eq:cost_function}, despite its efficiency in scalability with the number of agents. Instead, \changed{to more accurately capture the robot-level execution under localization uncertainty,} each agent is assigned an individual velocity control input \(u_i \in \mathbb{R}^2\) that defines a uniform control field on a per-robot basis. 
This gives the optimization variable \(u = [u_1, u_2, \ldots, u_N]^T\) for \(N\) agents with the following RV-OBC formulation:
\begin{equation}
    \begin{aligned}
    \min_{u,\,s} &\quad \|u\|^2 + \gamma s\\
    \quad \text{s.t.} &\quad \alpha_v V(\rho) + \dot{V}(\rho,u) - s\leq 0\\
    &\quad \alpha_h h(\rho) + \dot{h}(\rho,u) \geq 0.
\end{aligned}
\label{eq:OptimizationBased_sum_fields}
\end{equation}
The main difference between RV-OBC \eqref{eq:OptimizationBased_sum_fields} and SV-OBC \eqref{eq:OptimizationBased} is that \eqref{eq:OptimizationBased_sum_fields} calculates its density evolution in a distributed sense, rather than one velocity field for all agents: 
\begin{equation}
    \frac{\partial \rho(r, t)}{\partial t} = \sum_{i=1}^{N} \frac{\partial \rho_i(r, t)}{\partial t} = \sum_{i=1}^{N} \bigg[ - u_i(r) \nabla  \rho_i(r, t)  + T_i \Delta \rho_i(r, t) \bigg].
    \label{eq:FP_equation_distributed}
\end{equation}

\begin{remark}
    As stated at the beginning of Section~\ref{sec3: math formulation}, the Fokker-Planck equation \eqref{eq:FP_equation} contains a diffusion coefficient \(T\) representing the random motion noise of a robot. We define a constant \(T = T_i = 0.3^2 c u_{\max}/ 2 = 0.045 c u_{\max}\), \changed{corresponding to a zero-mean Gaussian distribution whose probability mass is conservatively bounded such that \(99\%\) lies in \([-cu_{\max}, cu_{\max}]\), where the maximum admissible command \(u_{\max}\) is scaled by a constant \(c\) determined by physical hardware limits. Events outside this \(99 \%\) tail are considered physically impossible in practice. This bound is introduced as a modeling assumption to ensure the well-posedness of the continuous PDE formulation, and reflects the bounded nature of robot actuation within hardware limits, rather than a probabilistic safety guarantee over time.} 
    A variable \(T(u) = 0.045 u\) can more accurately represent the command-dependent motion noise. We decided to use a constant maximum noise for two reasons: it is computationally cheaper. It captures other forces, such as inertia in a quadcopter, which persists even if the command is zero. 
\end{remark}

\begin{proposition}
\label{prop:backstep}
Let $\mathcal{A}\subset \Omega$ and consider isotropic Gaussian \eqref{eq: rho def} centered at robot position \(x_i\) with identical \changed{inverse covariance (precision) matrices} \(\Sigma = \lambda I \) and swarm PDF \(\rho = \sum_i^N \rho_i \).  
The swarm density evolves by the Fokker-Planck equation \eqref{eq:FP_equation_distributed} with the barrier function \(h\) as defined in system \eqref{eq:OptimizationBased_sum_fields}.

We show the safety conditions in a worst-case situation, namely, at time $t_{k-1}$, all agents apply the maximum control $u_i^* = u_{\max} \hat n_i$ with
\(
n_i = \arg\min_{\|u\|= u_{\max}} -2\int_A \rho\,u^\top \nabla \rho_i\,dr,
\)
as the solution to \eqref{eq:OptimizationBased_sum_fields} with saturated CBF, \(\alpha_h h_{tk-1} + \dot h_{tk-1} = 0\).
Then, by the motion noise definition, \(T = 0.045 c u_{\max}\), an upper bound for the worst-case noisy motion is
\(
\delta_i = (1+c) \Delta t\,u_{\max}\hat n_i = \kappa \hat n_i,
\)
with timestep size \(\Delta t\).

We denote \(e_i = r_i - x_i\), and the moments at time $t_{k-1}$:
\begin{equation}
    M_{0,i} := \int_A \rho\rho_i dr, \,
M_{1,i} := \int_A \rho \rho_i \hat n_i^T  e_idr,  \,
M_{2,i} := \int_A \rho \rho_i \|e_i\|^2 dr,
\end{equation}
and the aggregates
\begin{equation}
M_0 := \sum_i M_{0,i}, 
\qquad M_1 := \sum_i M_{1,i},
\qquad M_2 := \sum_i M_{2,i}.
\end{equation}
Then, we can define \(\beta_i := \exp\!\big(\lambda e_i^T \delta_i - \tfrac12 \delta_i^T \lambda \delta_i \big),\) and a weighted average
\begin{equation}
    \beta_{0, i,j} :=     \frac{ \int_A \beta_i(r)\,\rho_i(r)\,\beta_j(r) \rho_j(r)\,dr}
         { \int_A \rho_i(r)\,\rho_j(r)\,dr},
\end{equation}
with a single agent swarm level notation: 
\begin{equation}
    \bar\beta_{0,i}
    := \sum_j^N 
    \frac{ \int_A \beta_j(r)\,\rho_j(r)\,\beta_i(r)\rho_i(r)\,dr}
         { \int_A \rho_j(r)\,\rho_i(r)\,dr},
         \label{eq: beta bar def}
\end{equation}
with the denominator sums to \(M_{0,i}\), and aggregates to \(\bar \beta_0 = \sum_i \bar \beta_{0,i}\). Similarly, we define \(\bar \beta_{1,i}, \bar \beta_1, \bar \beta_{2,i}, \bar \beta_2\) as the weighted average with \(M_{1,i}, M_1, M_{2,i}, M_2\), respectively. 

Then, a sufficient condition for the CBF condition
\[
\alpha_h h_{tk} + \dot h_{tk} \;\ge\; 0
\]
to hold with a recovery command \(u_{i,tk} = - u_{\max} \hat n_i\) is:
\begin{equation}
    \begin{aligned}
        &(\bar \beta_1 +1) u_{\max} M_1 \geq \, T \lambda (\bar \beta_2 - 1)M_{2}- 2 T \bar \beta_1 \lambda \kappa M_1\\
        &\quad \quad \quad + \Big[\bar \beta_0 \big[T \lambda \kappa^2 + u_{\max}  \kappa - 2T + \frac{\alpha_h}{2 \lambda}\big] + 2 T - \frac{\alpha_h}{2 \lambda} \Big] M_0
    \end{aligned}
\label{eq: safety final bound}
\end{equation}
\end{proposition}

\begin{proof}
Since the swarm is moving towards \(\mathcal{A}\) with a saturated CBF at the previous timestep, we have:
\begin{equation}
    \alpha_h(h_{tk}-h_{tk-1}) + (\dot h_{tk} - \dot h_{tk-1}) \geq 0,
\label{eq: cbf goal prop}
\end{equation}
as the equivalent inequality goal for safety. 
By the identities in \cite{RalPaper}:
\begin{equation}
    \begin{aligned}
        \nabla\rho &= -\rho \Sigma (r - x) \\
        T \Delta \rho &= T \Bigg[\frac{\partial^2}{\partial r_1 ^2} \rho + \frac{\partial^2}{\partial r_2 ^2} \rho \Bigg]\\
        &= T \rho\big[ (r - x)^T \Sigma^2 (r - x) - \text{Tr}(\Sigma) \big],
    \end{aligned}
\end{equation}
the Fokker-Planck PDE at timestep \(t_k\) can be expressed as:
\begin{equation}
    \begin{aligned}
    \partial_t \rho_i = \rho_i(u_i^T \Sigma e_i + T(\lambda^2 e_i^T e_i - 2 \lambda)).
    \end{aligned}
\end{equation}
This yields a swarm-level advection gain, with recovery commands \(u_{i,tk} = - u_i^* \):
\begin{equation}
    \dot h_{tk}^{adv} = \sum_i^N -2 \int_\mathcal{A} \rho_{tk} \, \rho_{i,tk} \, (-u_{\max} \hat n_i)^T \Sigma (e_i - \delta_i) \, dr,
\end{equation}
and a diffusion loss:
\begin{equation}
    \dot h_{tk}^{diff} = \sum_i^N -2 T \int_\mathcal{A} \rho_{tk} \, \rho_{i,tk} \, (\lambda^2 (e_i - \delta_i)^T (e_i - \delta_i) - 2 \lambda) \, dr.
\end{equation}
Notably, as shown in Fig.~\ref{fig:BlockDiagram Optimization}, the PDF resets at each timestep, with new robot centers at \(x_i + \delta_i\), giving: 
\[
\rho_{i,tk} := \exp(- \frac{1}{2} (e_i - \delta_i)^T \Sigma (e_i - \delta_i)) = \rho_{i,tk-1} \, \beta_i.
\]
Further, by definition \eqref{eq: beta bar def}, we have 
\[
\int_\mathcal{A} \rho_{tk} \, \rho_{i,tk} \, dr = \bar \beta_{0,i} \int_\mathcal{A} \rho_{tk-1} \, \rho_{i, tk-1} \, dr = \bar \beta_{0,i} M_{0,i}
\]

Under the above definitions, dropping the timestep subscript as all functions are now in \(t_{k-1}\), the following advection gain for \(\dot h_{tk} - \dot h_{tk-1}\) is found:
\begin{equation}
    \begin{aligned}
       &= \sum_i^N 2 (\bar \beta_{1,i} + 1) \int_\mathcal{A} \rho \, \rho_{i} \, (u_{\max} \hat n_i)^T \Sigma e_i \, dr - 2 \bar \beta_{0,i} \int_\mathcal{A}\rho \, \rho_{i} \, (u_i^*)^T \Sigma \delta_i \, dr \\
       &=  2 (\bar \beta_1 +1) \lambda u_{\max} M_1 - 2 \bar \beta_0 u_{\max} \lambda \kappa M_0.
    \end{aligned}
\end{equation} 
Similarly, the diffusion penalty is:
\begin{equation}
    \begin{aligned}
        &= \sum_i^N -2 T \int_\mathcal{A} \rho (\sum_i^N \beta_i) \, \rho_{i} \beta_i\, (\lambda^2 (e_i - \delta_i)^T (e_i - \delta_i) - 2 \lambda) \, dr \\
        & \quad + 2 T \int_\mathcal{A} \rho \, \rho_{i} \, (\lambda^2 e_i^T e_i - 2 \lambda) \, dr \\
        &= -2T \lambda \Big[\lambda (\bar \beta_2 - 1)M_{2} - 2 (\bar \beta_0 - 1)M_{0} - 2 \bar \beta_1 \lambda \kappa M_1 + \bar \beta_0 \lambda \kappa^2 M_{0}  \Big].
    \end{aligned}
\end{equation}
Finally, the penalty on the barrier function itself is: 
\begin{equation}
    \begin{aligned}
        &= -\alpha_h \int_\mathcal{A} \rho_{tk}^2 - \rho_{tk-1}^2 \, dr = -\alpha_h (\bar \beta_0 - 1) M_0.
    \end{aligned}
\end{equation}
Substituting these into \eqref{eq: cbf goal prop} yields \eqref{eq: safety final bound}.

\end{proof}

Proposition~\ref{prop:backstep} guarantees the safety of the swarm under the bound \eqref{eq: safety final bound}. We look at two special cases to further understand what this bound implies. 
First, we consider a zero motion noise situation, \(c = 0\). Then, \eqref{eq: safety final bound} simplifies to:
\begin{equation}
    \begin{aligned}
        (\bar \beta_1 +1) u_{\max} M_1
        \geq \, \Big[\bar \beta_0 u_{\max}  \kappa + (\bar \beta_0 - 1)  \frac{\alpha_h}{2 \lambda} \Big] M_0 \\
         \approx \frac{\alpha_h}{2 \lambda} \Delta t u_{\max} M_1 + \bar \beta_0 \Delta t u_{\max}^2 M_0\\
         (\bar \beta_1 +1 ) M_1 \geq   \frac{\alpha_h}{2 \lambda} \Delta t M_1 + \bar \beta_0\Delta t u_{\max} M_0,
    \end{aligned}
\end{equation}
using a first-order approximation in \(h_{tk} - h_{tk -1}\). This simplified bound states that, to ensure safety, the net gain in advection needs to be greater than the safety drop and a penalty term that depends on the displacement from the last timestep. As a result, with \(\Delta t \rightarrow 0\), the bound is always satisfied, and \(4 \lambda \geq \alpha_h \Delta t\) is a sufficient bound as \(u_{\max} \rightarrow 0\). As expected, with perfect motion, safety is easily guaranteed with smaller timestep for fast correction, smaller \(\alpha_h\) to detect CBF saturation earlier, and lower maximum command and localization noise to decrease the risk of safety violation.

Then, we find an approximation of \eqref{eq: safety final bound} to demonstrate the effects of motion noise in the system. Notice that all \(\bar \beta_i \rightarrow 1^+\) with small step approximations \(\delta_i \rightarrow 0^+\) from \(c, \Delta t,\) or \(u_{\max}\), this gives:
\begin{equation}
    \begin{aligned}
        2 \Big[1 + 0.045c \lambda \kappa \Big]M_1 &\geq         (1+c) \Delta t u_{\max} \big[1 +  0.045 c  \lambda \kappa \big] M_0\\
        2 M_1 &\geq (1+c) \Delta t u_{\max}  M_0.
    \end{aligned}
\end{equation}
This simplified bound captures each parameter's effects on safety. A higher maximum control, \(u_{\max}\), increases the RHS, showing that a faster agent is harder to keep safe. Conversely, a smaller timestep \(\Delta t\) reduces the RHS, indicating that quicker correction enforces safety. 
The geometry of \(\mathcal{A}\) and \changed{precision} \(\lambda\) are encoded in \(M_1, M_0\). As danger areas expand outwards, the weighted distance term \(\hat n_i^T e_i\) drives \(M_1 > M_0\), indicating that greater weighted distance from \(\mathcal{A}\) improves safety. A more accurate localization means \(\lambda\) amplifies both \(M_0, M_1\) by concentrating the underlying PDFs. Paradoxically, this implies that a more accurate localization can make the bound harder to satisfy, since \(x_i\) of a sharper PDF would be closer to \(\mathcal{A}\), and reduce safety faster by the same motion noise. This outcome contradicts the noise free analysis, yet is consistent with the relationship between motion and localization noises.
Finally, as also shown in the noise-free case, a smaller motion noise coefficient, \(c\), relaxes the bound, making it easier to enforce safety. 

The bound in Proposition~\ref{prop:backstep} was evaluated in both simulations and experiments, and in all scenarios, the inequality \eqref{eq: safety final bound} was satisfied. Violations were observed only under extreme, artificial edge cases in which nearly all robots, under heavy motion noise, simultaneously applied large maximum magnitude \(u_{\max}\) directed towards \(\mathcal{A}\) in close distance, with \(\Delta t \geq 1\)~s. Outside of these extreme scenarios, the bound \eqref{eq: safety final bound} was consistently held in practice.

\begin{proposition}
    The constrained optimization problem \eqref{eq:OptimizationBased_sum_fields} returns a unique optimal solution \(u^*\), given that the initial swarm distribution is safe and parameters can satisfy \eqref{eq: safety final bound}. This \(u^*\) ensures the forward invariance of the set \(\{ \rho : h(\rho)\ge0\}\).
\end{proposition}
\begin{proof}
    By Proposition~\ref{prop:backstep}, the CBF constraint (second constraint in \eqref{eq:OptimizationBased_sum_fields}) is feasible under bound \eqref{eq: safety final bound}. Moreover, it ensures swarm safety recursively, since the bound depends explicitly on the previous states. Additionally, the CLF constraint (first constraint in \eqref{eq:OptimizationBased_sum_fields}) is always feasible with the slack variable \(s\). Finally, as the optimization cost is strictly convex, the optimal solution \(u^*\) is unique.
\end{proof}

While this controller formulation ensures forward invariance under a safe initial distribution with conditions based on predefined parameters, the computation time increases drastically with the size of the swarm. \changed{This is because each additional agent requires computing an additional full spatial field, in addition to the two optimization variables.} To improve the computation speed, we propose the following Voronoi inspired OBC. 

\subsection{Robot Velocity Optimization-Based Control with Voronoi Partitions (RV-OBC-V)}
\label{subsec: RV-OBC-V}
Voronoi-based controllers have been studied for coverage control in the ODE domain thanks to their fast online computation capabilities \cite{doi:10.1177/0278364908100177,TERUEL201951, Elamvazhuthi_Berman_2019}.
In the below formulation, instead of accounting for an entire new field with each additional agent as seen in RV-OBC \eqref{eq:OptimizationBased_sum_fields}, a Voronoi partition step is added such that each agent only accounts for its density within its own cell. Specifically, we define the Voronoi cells centered at each agent position, \(x_i\), as a disjoint partition of the spatial field: \(\bigcup_{i=1}^N C_i = \Omega, \ \ C_i \cap C_j = \varnothing \;\; \forall i \neq j\). The resulting optimization-based controller is:
\begin{equation}
    \begin{aligned}
    \min_{u,\,s} &\quad \|u\|^2 + \gamma s\\
    \quad \text{s.t.} &\quad \alpha_v V_c(\rho) + \dot{V}_c(\rho,u) - s\leq 0\\
    &\quad \alpha_h h_c(\rho) + \dot{h}_c(\rho,u) \geq 0,
\end{aligned}
\label{eq: RV-OBC-V}
\end{equation}
with the Voronoi versions of the Lyapunov and barrier functions:
\begin{equation}
    V_C = \sum_i^N \frac{1}{2} \int_{C_i} (\rho_d - \rho_i)^2 \,dr \quad h_C = \epsilon - \sum_i^N \int_{C_i \bigcap \mathcal{A}}\rho_i^2\,dr,
\end{equation}
and their respective time derivatives. This new formulation, while avoiding the time and space scaling problem associated with increasing numbers of agents, does not account for densities outside of each agent's cell. In the remainder of this section, we derive an upper bound to quantify this overestimated amount. 

First, we find the difference in the Lyapunov measurement of target density matching:
\begin{equation}
    \begin{aligned}
        V_c - V &= \sum_i^N  \frac{1}{2} \int_{C_i} (\rho_d - \rho_i)^2 \,dr - \frac{1}{2} \int_\Omega (\rho_d - \sum_i^N \rho_i )^2 \ dr\\
        &= \frac{1}{2} \sum_i^N \int_{C_i} (\rho_d - \rho_i)^2 - (\rho_d - \sum_i^N \rho_i )^2 \ dr \\
        &= \frac{1}{2} \sum_i^N \int_{C_i} 2 \rho_d (\sum_i^N \rho_i - \rho_i)+ \rho_i^2 - (\sum_i^N \rho_i)^2  \ dr.
    \end{aligned}
    \label{eq: V-Vc first}
\end{equation}
The last two terms can be simplified as:
\begin{equation}
    \rho_i^2 - (\sum_i^N \rho_i)^2 = \rho_i^2 - (\sum_{k \neq i}^N \rho_k + \rho_i)^2 = -(\sum_{k\neq i}^N \rho_k)^2 - 2\rho_i \sum_{k\neq i}^N \rho_k,
\end{equation}
which means \eqref{eq: V-Vc first} becomes:
\begin{equation}
    \begin{aligned}
        V_c - V &= \frac{1}{2} \sum_i^N \int_{C_i} 2 \rho_d (\sum_{k\neq i}^N \rho_k) - (\sum_{k\neq i}^N \rho_k)^2 - 2\rho_i \sum_{k\neq i}^N \rho_k \ dr\\
        &= \frac{1}{2} \sum_i^N \int_{C_i} -(\sum_{k\neq i}^N \rho_k)^2 + 2(\rho_d - \rho_i) \sum_{k\neq i}^N \rho_k \ dr.
    \end{aligned}
    \label{eq: V-Vc main}
\end{equation}
Then, to find a bound on \eqref{eq: V-Vc main}, we first define a few parameters. Let \(C_i^c = \Omega \setminus C_i\) be the complement of each cell, and \(S_i = \int_{{C_i^c}} \rho_i \ dr\) be the mass of agent \(i\)'s density with \(s_i = \sup_{x \in C_i^c} \rho_i(x)\) as the upper density bound outside each cell. By Holder's and triangle inequality, the second term in \eqref{eq: V-Vc main} is bounded as:
\begin{equation}
\begin{aligned}
    \int_{C_i} 2(\rho_d - \rho_i) \sum_{k\neq i}^N \rho_k \ dr & \leq 2 \sup_{x \in C_i} \bigg|\sum_{k\neq i}^N \rho_k \bigg| \int_{C_i} \big| \rho_d - \rho_i \big| \ dr \\
    &\leq 2 \sum_{k \neq i}^N s_k \int_{C_i} \big| \rho_d - \rho_i \big| \ dr.
\end{aligned}
\end{equation}
Putting this together, \eqref{eq: V-Vc main} simplifies to:
\begin{equation}
    \begin{aligned}
        V_c - V &\leq \sum_i^N \bigg[ \sum_{k \neq i}^N s_k  \int_{C_i} \big| \rho_d - \rho_i \big| \ dr \bigg] \\
        &= \sum_i^N \bigg[ \sum_{k \neq i}^N s_k  E_i\bigg],
    \end{aligned}
    \label{eq: final V-Vc bound}
\end{equation}
where \(E_i = \int_{C_i} \big| \rho_d - \rho_i \big| \ dr\) represents the error in matching the target density within \(C_i\). 

Similarly, for the barrier function difference, we have:
\begin{equation}
    \begin{aligned}
        h_c - h &= - \sum_i^N \int_{C_i \bigcap \mathcal{A}} \rho_i^2 \, dr + \int_\mathcal{A} (\sum_i^N \rho_i)^2 \, dr \\
        &= \sum_i^N \int_{C_i \bigcap \mathcal{A}} (\sum_i^N \rho_i)^2 - \rho_i^2 \, dr \\
        &= \sum_i^N \int_{C_i \bigcap \mathcal{A}} (\sum_{k\neq i}^N \rho_k)^2 + 2 \rho_i (\sum_{k\neq i}^N \rho_k) \, dr\\
        & \leq \sum_i^N \bigg[ \sum_{k \neq i}^N s_k \sum_{k \neq i}^N S_k + 2 \sum_{k \neq i}^N s_k \int_{C_i \bigcap \mathcal{A}} \rho_i \, dr \bigg].
    \end{aligned}
    \label{eq: hc h ineq}
\end{equation}

\begin{remark} \label{Remark: hc - h tiny}
    It is important to note that both \eqref{eq: final V-Vc bound} and \eqref{eq: hc h ineq} have a loose bound that can be tightened significantly based on the size of \(C_i\)s, and \(\mathcal{A} \) with respect to \(\Omega\). The terms \(s_k, \ S_k\) can be further bounded within  \(C_i \bigcap C_k^c\), \(\mathcal{A} \bigcap C_k^c\) instead of \(C_k^c\). Furthermore, for \eqref{eq: hc h ineq}, many of the integral terms vanish in practice since \(\mathcal{A}\) is typically not covered by all Voronoi cells. As shown in Section~\ref{section:Simulation and Results}, the difference between the barrier functions is trivial compared to Lyapunov functions, and can be added as an additional \(\epsilon_c\) for the CBF constraint if needed. 
\end{remark}

Therefore, it is clear to see from \eqref{eq: final V-Vc bound} and \eqref{eq: hc h ineq} that the Voronoi cell approach converges towards the true measurements in \eqref{eq:OptimizationBased_sum_fields} as \(s_i\) and \(S_i\) decrease. Both of which are achieved by agents working together but further apart. Although the term \(E_i\) in \eqref{eq: final V-Vc bound} can increase with larger \(C_i\), it converges towards a constant bounded by the target and agent's total density mass in an additive fashion through the integral. Its multiplier, \(s_i\), decreases exponentially with respect to the increase of \(C_i\) and dominates \(E_i\) such that \eqref{eq: final V-Vc bound} converges to zero. At the extreme, where all agents' PDFs are disjoint, the Voronoi cell approach has the same values as \eqref{eq:OptimizationBased_sum_fields} while computing at a much faster rate. Additionally, this formulation also becomes more accurate with a smaller number of agents as a result of the summations in the bounds and sizes of \(C_i\)s. 

\begin{center}
    \vspace{0.5\baselineskip}
    \fbox{%
        \begin{minipage}{0.95\linewidth}
            We conclude that our proposed algorithm RV-OBC-V does scale with respect to the swarm size, with minimal overhead due to the computation of the Voronoi partition. Contrary to a common swarm stereotype, however, it works best with sparse swarms, \changed{where the lack of connectivity among the agents makes the satisfaction of invariance constraints more effective}. In other words, \textit{it is much safer to be sparse than connected}.
        \end{minipage}
    }
\end{center}

\section{Simulation Results and Analysis}
\label{section:Simulation and Results}
To test the proposed controller, we conducted simulations within \changed{a square domain with side length \(4\)~m,} discretized at a spatial resolution of \(0.1\)~m, resulting in a grid of \(41 \times 41 = 1681\) points. A periodic boundary condition is applied to improve the stability of OC. The target PDF is a Gaussian centered at \((0.5, -0.5)\), with a standard deviation of \(1.5\). Simulations were conducted over a total duration of \(4\)~s, with \(\Delta t =0.01\)~s. 

\subsection{Single-robot Study}
The comparison in the following is conducted using a single-robot setup to highlight the distinct behaviors of the two controllers, OC \eqref{eq:Optimal_System} and SV-OBC \eqref{eq:OptimizationBased_discretized}, with the simulation setup in Fig.~\ref{fig: Simulation Setup}. Multi-robot simulations are discussed in the next subsection to validate the controllers' performance in a team of robots. To replicate real-world noisy behavior, we implemented both localization and motion noise in all controllers according to \changed{Algorithm~\ref{alg: noise}}. To achieve statistically reliable results, the simulations were performed \(100\) times, and the averaged performance metrics were computed. 

\begin{figure}
    \centering
    \includegraphics[width=0.6\linewidth]{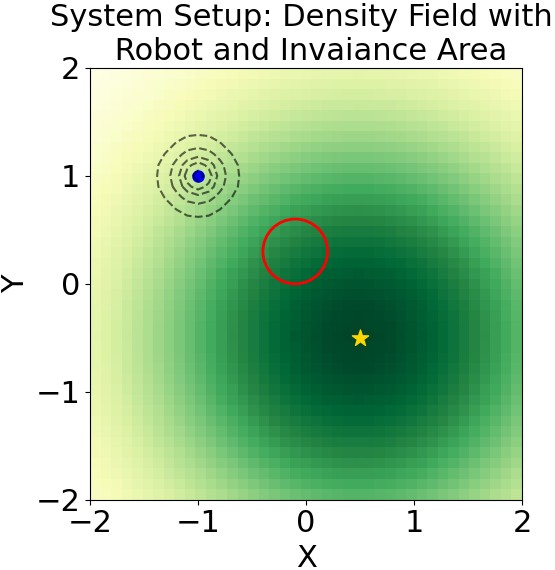}
    \caption{Simulation setup: the red circle represents the unsafe area, the blue dot marks the robot’s position with the black dashed contour as its PDF, and the yellow star is the target center with green PDF.}
    \label{fig: Simulation Setup}
\end{figure}

\begin{algorithm}
\caption{\enskip Simulation Flow}\label{alg1}
\begin{algorithmic}
  \State Initialize \(x_i^{true}, x_i^{meas}\) for each agent
  \For {each timestep $t_k$}
  \State draw \(x_i^{meas}\) from a Gaussian distribution around \(x_i^{true}\)
  \State compute \(u_i^*\) from \(x_i^{meas}\) using one of \eqref{eq:Optimal_System}, \eqref{eq:OptimizationBased_discretized}, \eqref{eq:OptimizationBased_sum_fields}, or \eqref{eq: RV-OBC-V}.
  \State update \(x_i^{true}\) with \(u_i^*\) with added Gaussian motion noise
  \EndFor
\end{algorithmic}
\label{alg: noise}
\end{algorithm}

An interesting aspect revealed by OC is the reshaping of the robot PDF under the state invariance constraint, defined as maintaining the total density within a circular region below a threshold. Though intuitively reasonable, this constraint led the controller to adopt the classic optimal transport behavior of allowing the PDF to flatten itself by diffusion to save control cost, only later reshaping it into the desired Gaussian-like form, seen by the spreading effect in the left column of Fig.~\ref{fig: Single Time Sequence}. As a result, the final peak reaches only about \(10\%\) of the initial Gaussian peak despite maintaining the correct shape.
\changed{
This behavior arises because OC interprets the swarm probability density as a continuously evolving belief governed solely by the Fokker-Planck equation, whereas in practice, the robot’s PDF is repeatedly reset by new localization measurements.
As seen in Fig.~\ref{fig:BlockDiagram Optimization}, a new location measurement is taken at each timestep, effectively re-centering the robot's PDF at the measured location. In contrast, Fig.~\ref{fig:BlockDiagram Optimal} showcases that OC evolves only with the initial density, omitting this unavoidable PDF reset in applications.
Consequently, OC predicts much lower density inside \(\mathcal{A}\) than reality, since the PDF is not continuously re-centered at the robot's location. This leads to a false sense of safety near \(\partial \mathcal{A}\), as shown in dark green in the middle left plot of Fig.~\ref{fig: Single Time Sequence}, where the true PDF mass would be much higher.}
Reducing the constraint tolerance \(\epsilon\) seems intuitive, yet numerical experiments reveal that very small values of \(\epsilon\) increase numerical instability with limited improvements, as the reshaping effect grows with time. 

\begin{figure}
    \centering
    \subfloat{\includegraphics[width=0.47\linewidth]{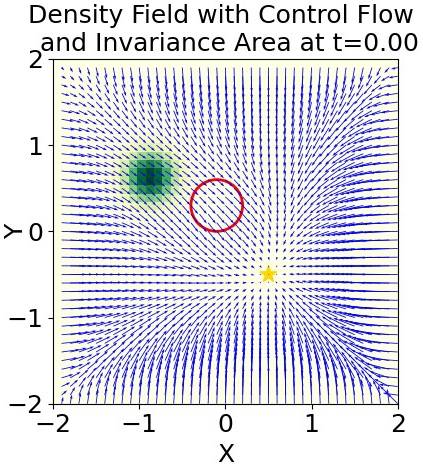}}
    \subfloat{\includegraphics[width=0.47\linewidth]{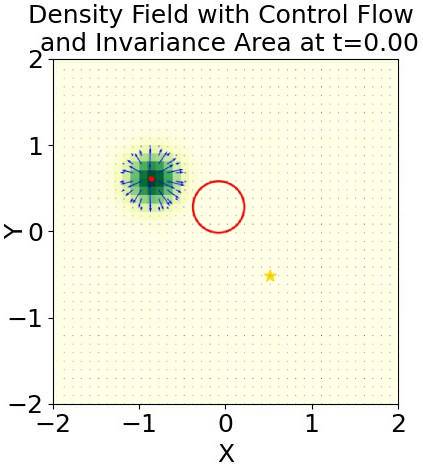}}\\    
    \subfloat{\includegraphics[width=0.47\linewidth]{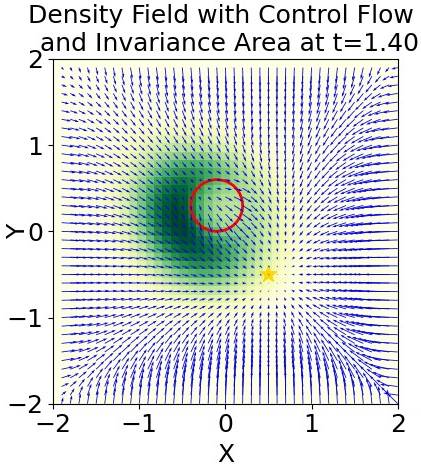}}    
    \subfloat{\includegraphics[width=0.47\linewidth]{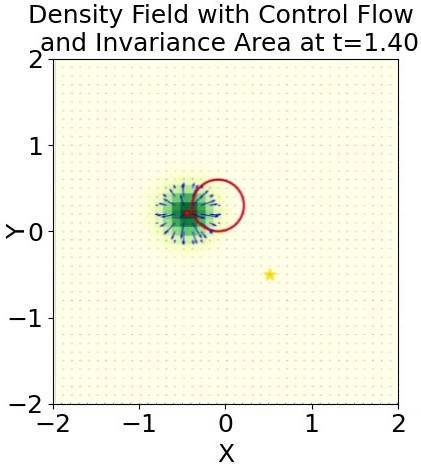}}\\    
    \subfloat{\includegraphics[width=0.47\linewidth]{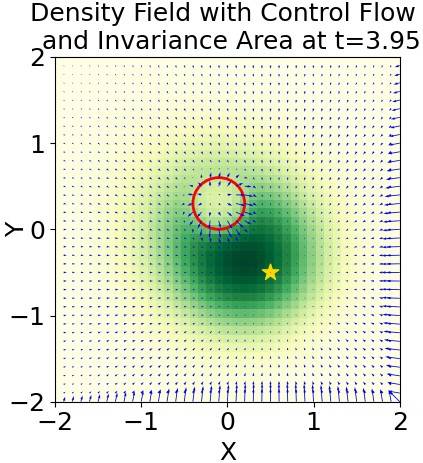}}    
    \subfloat{\includegraphics[width=0.47\linewidth]{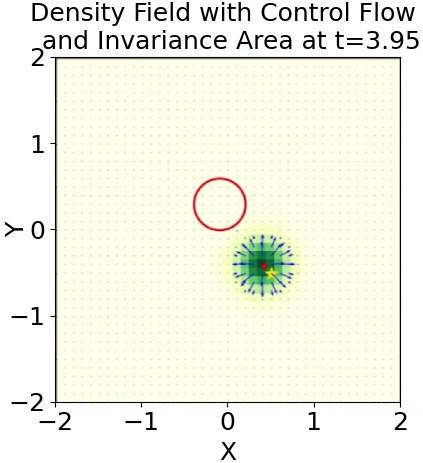}}
    \caption{Time sequence of OC (left column) and SV-OBC (right column). Star is the mean of the target, blue is the control field, the red circle represents the unsafe area, and green is robot density (darker is higher) with a red dot as the robot position.}
    \label{fig: Single Time Sequence}
\end{figure}
In contrast, \changed{the OBCs adopt a fundamentally different strategy, as they optimize the control input for just one time step ahead.} Specifically, SV-OBC generates a control field that is nearly zero everywhere except in the region immediately around the robot, where it radially pushes the PDF outward to flatten it towards the target, as seen in the right column of Fig.~\ref{fig: Single Time Sequence}. This approach yields lower commanded control effort by interpolating around the agent location.
RV-OBC, by design, uses uniform control fields characterized by \(u_i^* \in \mathbb{R}^2\). This formulation provides a more accurate description of the noisy system and ensures safety at the tradeoff of potentially higher control magnitude. 

The distinction between the three approaches in control efforts is shown in the top plot of Fig.~\ref{fig: performance plot-single robot}. OC has the lowest control cost with a globally Lipschitz-continuous control field \cite{OptimalMain}, while SV-OBC applies aggressive, discontinuous control commands. It also has a lower steady-state control effort than RV-OBC by interpolating within radially outwards control fields. Notably, RV-OBC's near-constant control effort after \(1\)~s is the result of oscillation near the target center. At this stage, RV-OBC is expected to invest more effort as moving the entire PDF achieves greater CLF reductions than SV-OBC.  

\begin{figure}
    \centering
    \subfloat{\includegraphics[width=0.6\linewidth]{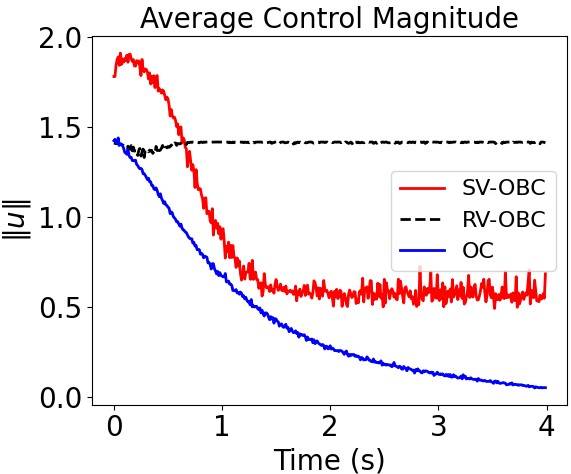}}\\    
    \subfloat{\includegraphics[width=0.6\linewidth]{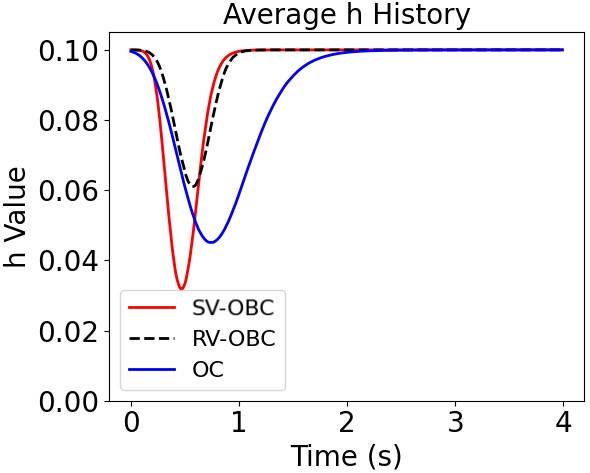}}\\    
    \subfloat{\includegraphics[width=0.6\linewidth]{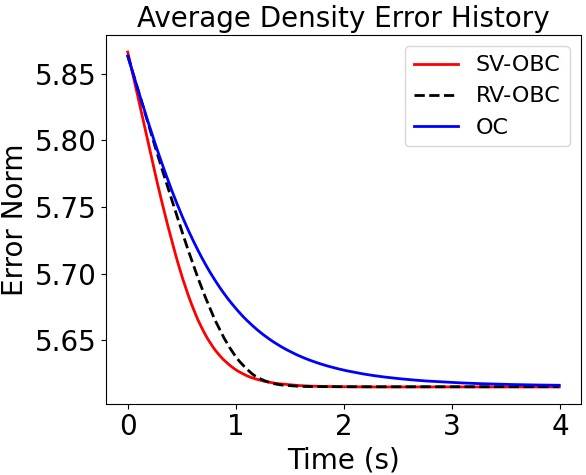}}    
    \caption{Single agent performance plots averaged with \(100\) simulations. Top: Control effort \(\|u\| \), lower means less energy spent. Middle: Safety barrier \(h\), higher means further from the danger area. Bottom: Lyapunov function \(\|\rho_d - \rho\|_{L^2(\Omega)}^2 \), lower means matching target better.}
    \label{fig: performance plot-single robot}
\end{figure}

The second plot in Fig.~\ref{fig: performance plot-single robot} shows the performance in maintaining area invariance, measured by the barrier function \(h\). As suggested by the control effort plot, SV-OBC employs the most discontinuous, aggressive adjustments in control directions, rapidly pushing the agent towards the invariance boundary if beneficial in the CLF. As a result, it has the lowest safety value out of the three controllers. However, with a smaller localization noise based PDF, it never violated the \(h\geq 0 \) criteria for forward invariance. Conversely, OC, while having a higher average safety measure, violated the safety constraint in \(19\) out of \(100\) simulations. Most violations are from the agents located too close to the boundary of \(\mathcal{A}\). As noted earlier, OC treats an agent at the boundary of \(\mathcal{A}\) as safe, but the safety of the measured position does not guarantee the safety of the true position due to resetting PDFs. Finally, since RV-OBC is designed with the most accurate relationship between measurement and true position, it has the highest average safety score with zero violations even in the presence of noise. 

Lastly, the final plot in Fig.~\ref{fig: performance plot-single robot} highlights the density matching error with a predefined target. Unsurprisingly, both OBCs outperformed OC with their more expensive control efforts spent. RV-OBC, while it converges to the target slightly slower, has the benefit of less aggressive control input and higher safety values than SV-OBC.

Therefore, just from single-agent behavior alone, we can conclude that OC is unsuitable for real-life safety-critical applications with resetting localization noise. 
While SV-OBC satisfied the safety threshold in these tests, it failed when the noise gets larger, as predicted in Fig.~\ref{fig: field control edgecase}. Therefore, in the next section, multi-robot simulations are run only for the safer and more accurate models, RV-OBC and RV-OBC-V.

\subsection{Multi-robot Study}
\begin{figure}
    \centering
    \subfloat{\includegraphics[width=0.6\linewidth]{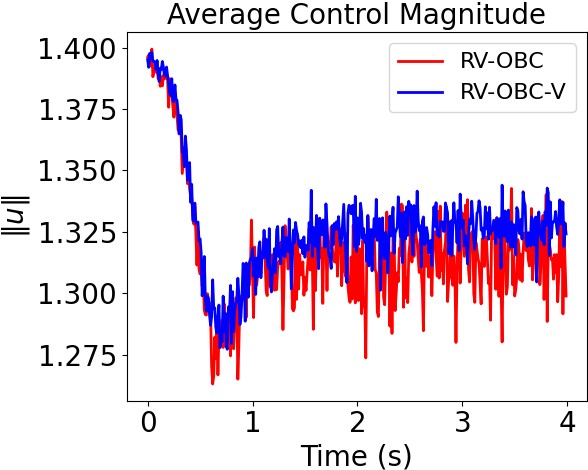}}\\    
    \subfloat{\includegraphics[width=0.6\linewidth]{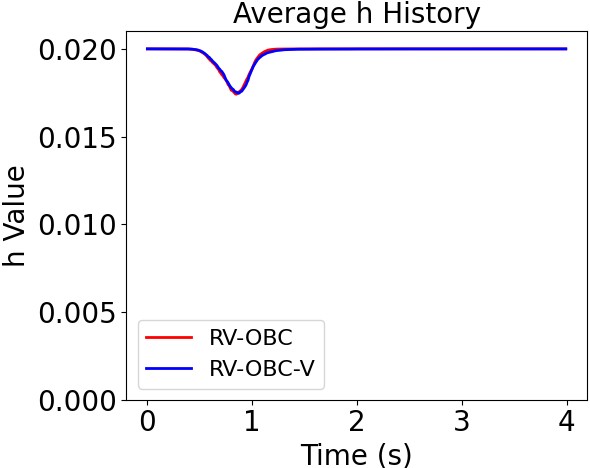}}\\    
    \subfloat{\includegraphics[width=0.6\linewidth]{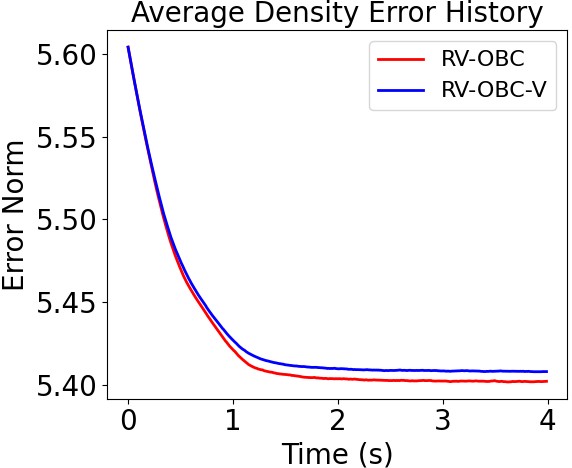}}    
    \caption{Six agent performance plots averaged with \(100\) simulations. Top: Control effort \(\|u\| \), lower means less energy spent. Middle: Safety barrier \(h\), higher means further from the danger area. Bottom: Lyapunov function \(\|\rho_d - \rho\|_{L^2(\Omega)}^2 \), lower means matching target better.}
    \label{fig: performance plot-6 robot}
\end{figure}

The proposed controllers are validated using multi-robot simulations to demonstrate their effectiveness in maintaining safety constraints while trying to match a target PDF. 
We take a closer look at a \(6\)-robot system, whose averaged performance evaluation over \(100\) runs is shown in Fig.~\ref{fig: performance plot-6 robot}. It is important to remember that RC-OBC-V is a Voronoi-based variant of the RC-OBC model aimed to increase computation speed with a trade-off for precise Lyapunov and barrier function evaluations, as discussed in Section~\ref{subsec: RV-OBC-V}.

From these plots, it is clear to see that RC-OBC-V applies slightly \changed{larger aggregated control magnitude over \(6\) robots} than the original RC-OBC, though in a less aggressive fashion, to achieve nearly identical safety results. In all \(100\) simulations, both controllers never violated the \(h \geq 0\) safety constraint, even under localization and motion noises. As expected, the difference, \(V_c - V\), is higher when the robots are closer together after converging to the target, causing a slightly higher equilibrium cost. In contrast, as stated in Remark~\ref{Remark: hc - h tiny}, the effect of the barrier function differences is negligible in comparison since fewer cells intersect with \(\mathcal{A}\). 

Then, we examine a simulation out of the \(100\) runs, as shown in the time sequence plots in \changed{Fig.~\ref{fig:6_robot_timeSequence}} and Supplementary Video~\(1\). These plots are set up with the target PDF in green and its mean at the gold star, \(\mathcal{A}\) in red, the measured positions as blue dots, and the true position trails in blue with black contour PDFs. 
The top row displays the RV-OBC result, while the bottom row shows the RV-OBC-V variant with cells partitioned in black lines.
\begin{figure*}[ht]
    \centering
    \subfloat{\includegraphics[width=0.165\linewidth]{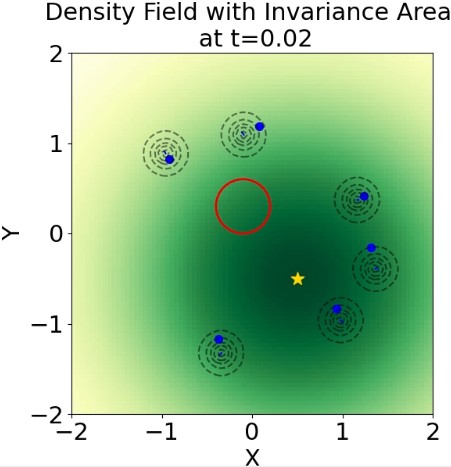}}%
    \subfloat{\includegraphics[width=0.165\linewidth]{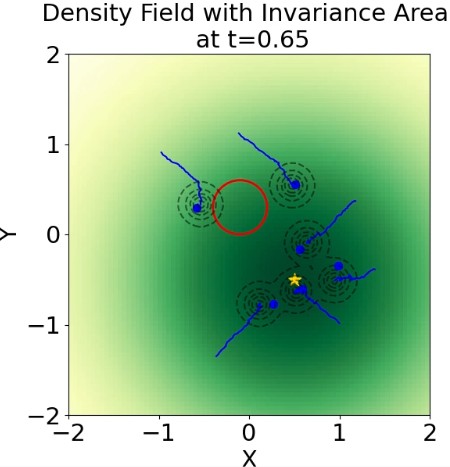}}%
    \subfloat{\includegraphics[width=0.165\linewidth]{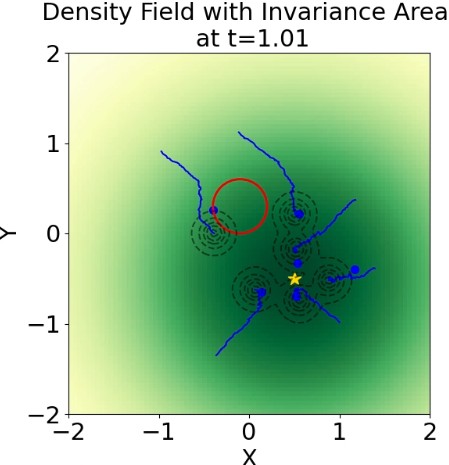}}%
    \subfloat{\includegraphics[width=0.165\linewidth]{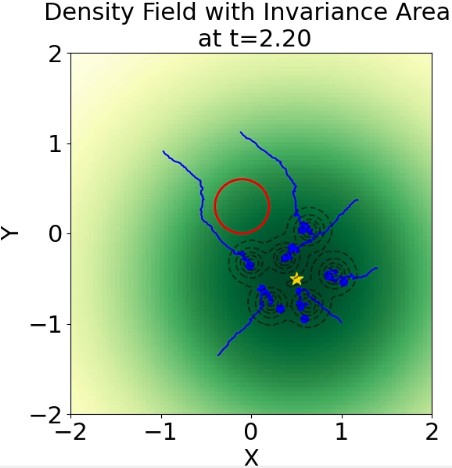}}%
    \subfloat{\includegraphics[width=0.165\linewidth]{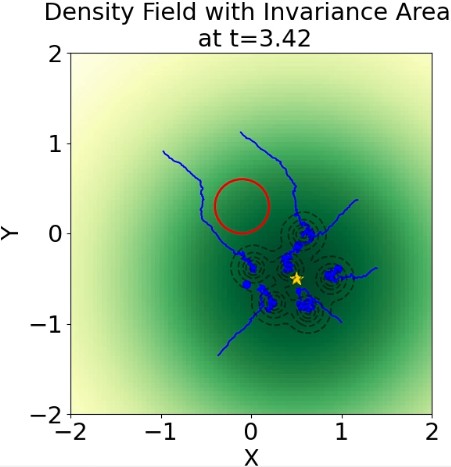}}%
    \subfloat{\includegraphics[width=0.165\linewidth]{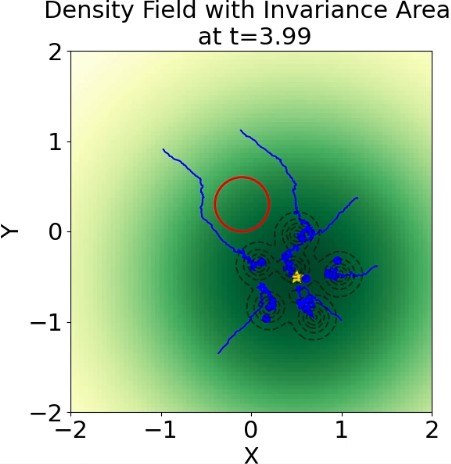}}\\    
    \subfloat{\includegraphics[width=0.165\linewidth]{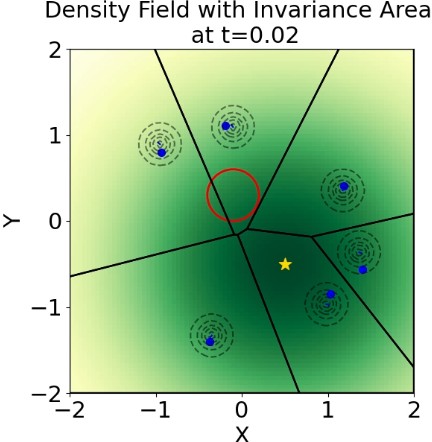}}%
    \subfloat{\includegraphics[width=0.165\linewidth]{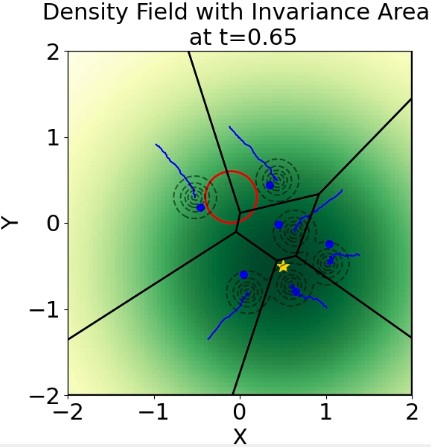}}%
    \subfloat{\includegraphics[width=0.165\linewidth]{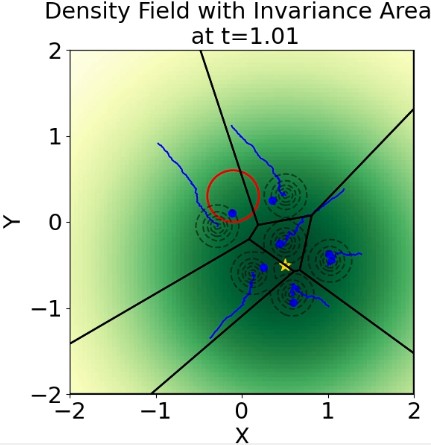}}%
    \subfloat{\includegraphics[width=0.165\linewidth]{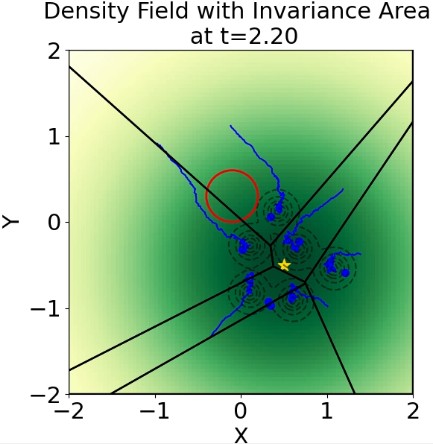}}%
    \subfloat{\includegraphics[width=0.165\linewidth]{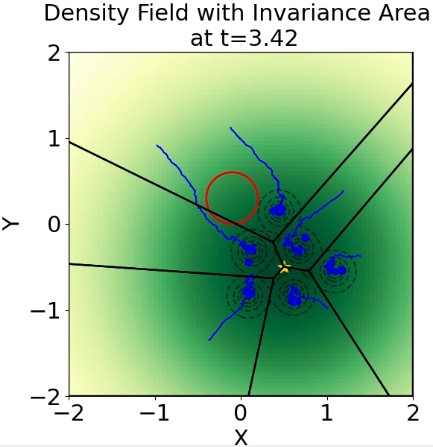}}%
    \subfloat{\includegraphics[width=0.165\linewidth]{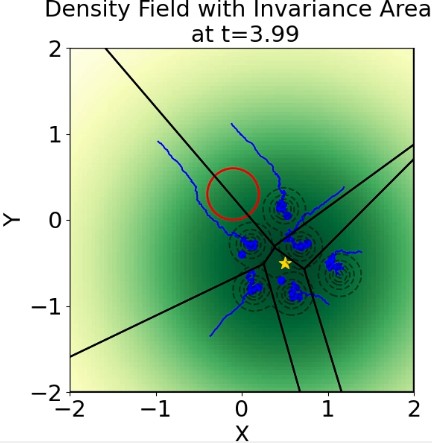}}

    \caption{Time sequence of RV-OBC (top row) and RV-OBC-V (bottom row). The star indicates the mean of the target, green is the target density (darker is higher), red is the invariance area, dark black line represents Voronoi cell boundaries, and blue represents robots' measured positions, with the true positions' trajectory trails and black PDF contour.}
    \label{fig:6_robot_timeSequence}
\end{figure*}

The RV-OBC converges smoothly towards the target, with the two top robots avoiding the danger area regardless of the noisy measurements. The robots quickly reach equilibrium near the target, with minor numerical oscillations. In contrast, the RV-OBC-V controller only considers the mass density of agents within its cell. With only noisy measurements available, the Voronoi partitions are often inaccurate, with cell boundaries crossing the PDFs, as seen in most plots. However, both of the top agents avoided the danger area and successfully converged towards the target, since they are sufficiently apart and cover \(\mathcal{A}\) with only two cells. After \(1\)~s, the robots cluster closely, and the inaccuracy in \(V_c\) began to affect the performance, as reflected in Fig.~\ref{fig: performance plot-6 robot}. 

\begin{remark}
    The safety performance of the RV-OBC-V controller is affected by the relation between cell size and the localization error defined by \(\Sigma\). One disruptive behavior arises when the cell boundary intersects with an agent's PDF, causing a nonzero probability that the measured position falls outside of the true cell. Moreover, in some cases, where two agents are so close together, their measured positions appear in each other's cells, effectively flipping the cell assignment and applying controls in the opposite directions, e.g., time \(2.20\)~s in Fig.~\ref{fig:6_robot_timeSequence}. While this effect might cause safety failures, it is only possible when agents are so close that their PDFs overlap and cover \(\mathcal{A}\) with a non-trivial amount. This further supports our claim that RV-OBC-V performs better when agents are further apart, challenging the stereotype that swarms are more effective in close proximity.  
\end{remark}

Finally, the runtimes of the two controllers are compared to validate the claim that RV-OBC-V significantly improves computation efficiency. \changed{Simulations with \(6, 10, 15\) and \(20\) robots are repeated for \(100\) times to obtain the average runtime per timestep with an AMD Ryzen $5$ $5600$ $6$-core processor, using the MOSEK Fusion API for Python \cite{mosek}.} Based on Fig.~\ref{fig: RuntimePlot}, with fewer agents, the overhead of Voronoi partition causes the RV-OBC-V to run slightly slower than RV-OBC. However, as the swarm size increases, the runtime for \changed{RV-OBC-V} stays nearly constant while RV-OBC grows rapidly. This difference arises from the structure of the controllers: RV-OBC must compute the full \(N_x \times N_y\) field for each new agent, whereas RV-OBC-V performs a light-weight partition step and optimizes in a single field regardless of swarm sizes. 

\begin{figure}
    \centering
    \subfloat{\includegraphics[width=0.7\linewidth]{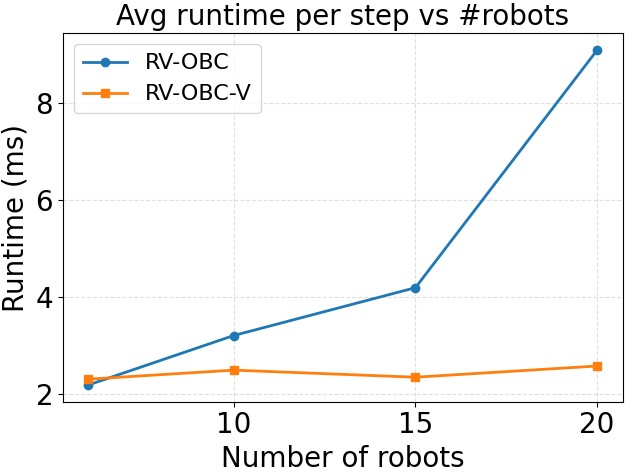}}    
    \caption{Runtime plot of RV-OBC and RV-OBC-V. The experiments were conducted on a \(81 \times 81\) grid and averaged over \(100\) runs each.}
    \label{fig: RuntimePlot}
\end{figure}

In summary, the simulations demonstrated that both controllers are able to maintain the safety constraint under localization and motion noise, with RC-OBC-V achieving comparable results with significantly reduced computational requirements. In the next section, we test the RC-OBC-V with physical robots to demonstrate its practical applicability, as it offers scalability over RC-OBC. 

\section{Experiment}
The proposed RC-OBC-V controller is implemented in a team of \(6\) DJI omnidirectional robots, \cite{DJI:RoboMasterS1UserManual}, \changed{in a square field with side length \(4\)~m.} The setup inherently includes real-world uncertainties: localization is achieved by a Vicon system, \cite{Vicon:NexusUserGuide}, subject to sensor noise, while the robots themselves are imperfect actuators, with friction and rapid directional changes introducing motion noise. The experiment was run for \(27.5\)~s and the resulting time sequence is shown in Fig.~\ref{fig: exp_robot_timeSequence} and Supplementary Video~\(2\).

As shown in the plots, the \(6\) robots are initially spread around the field, with the invariance area blocking half of the robots' direct paths to the target. Robots start converging towards the target directly for the first \(7\)~s, guided by the CLF, as robots are not close to \(\mathcal{A}\). As the top left three robots approach \(\mathcal{A}\), the resulting CBF-dominated control enforces their safe movement around the danger zone while approaching the target. Finally, as all robots are past the danger zone, some numerical oscillations are seen in the last few seconds as they converge towards an equilibrium distribution at the target location.

\begin{figure*}
    \centering
    \subfloat{\includegraphics[height=0.2\linewidth]{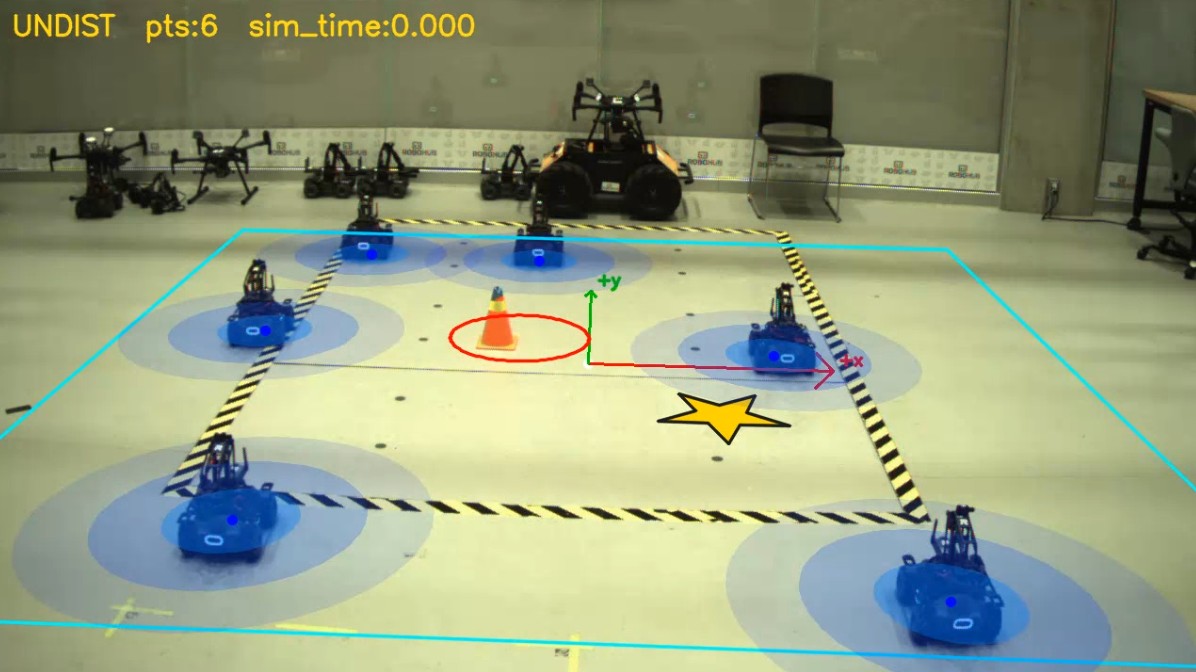}}\hspace{0.5cm}
    \subfloat{\includegraphics[height=0.2\linewidth]{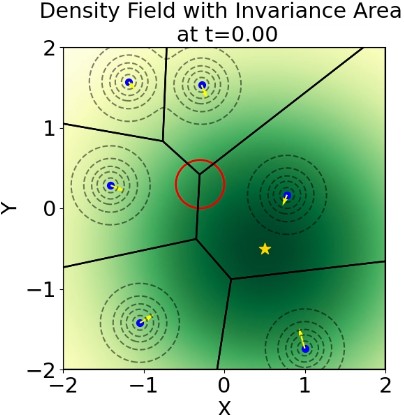}}\\
    \subfloat{\includegraphics[height=0.2\linewidth]{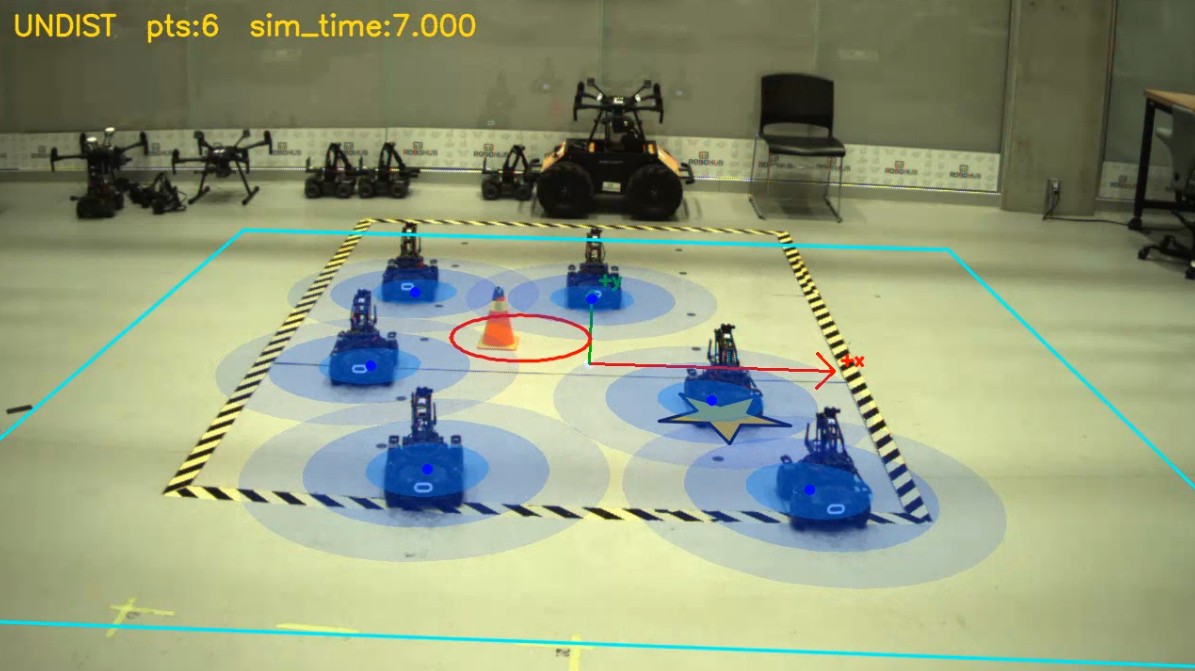}}\hspace{0.5cm}
    \subfloat{\includegraphics[height=0.2\linewidth]{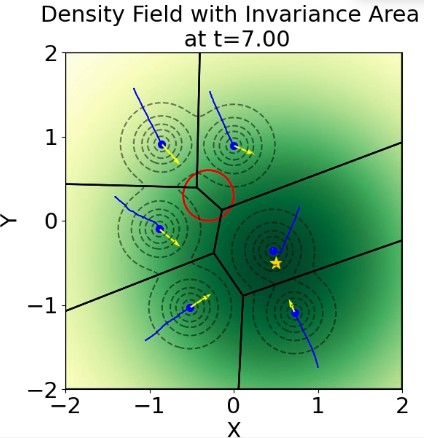}}\\
    \subfloat{\includegraphics[height=0.2\linewidth]{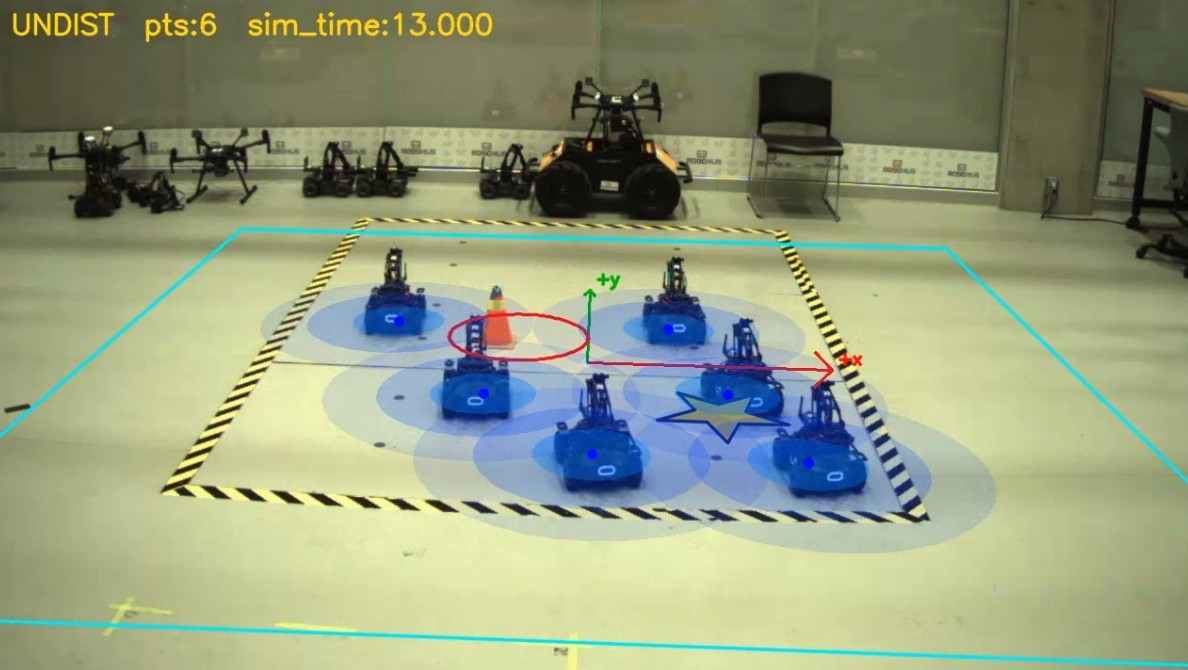}}\hspace{0.5cm}
    \subfloat{\includegraphics[height=0.2\linewidth]{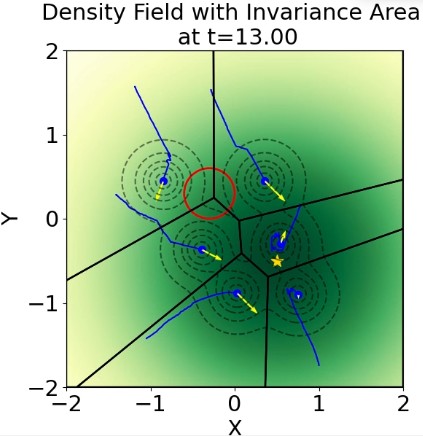}}\\
    \subfloat{\includegraphics[height=0.2\linewidth]{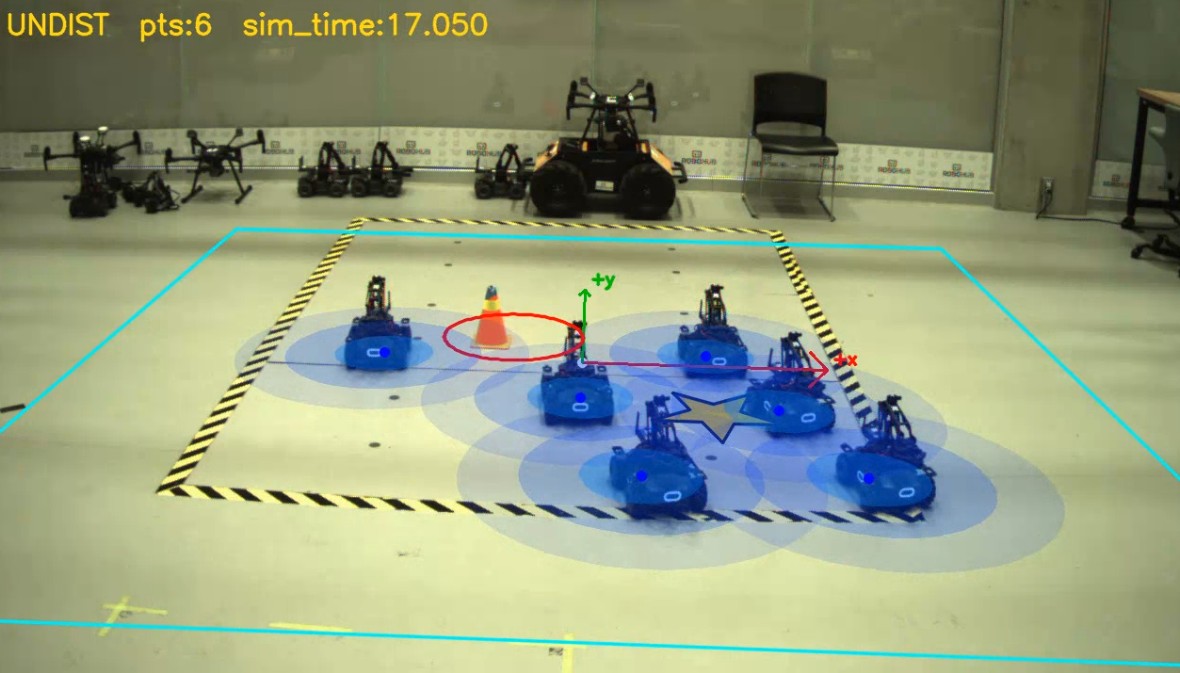}}\hspace{0.5cm}
    \subfloat{\includegraphics[height=0.2\linewidth]{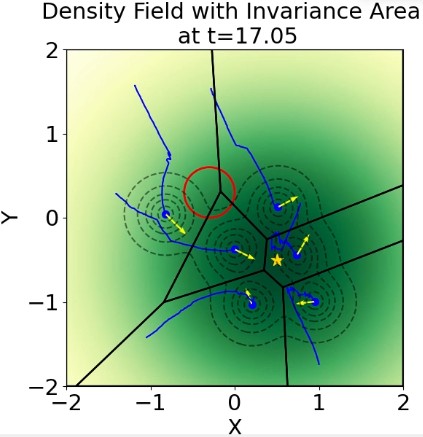}}\\
    \subfloat{\includegraphics[height=0.2\linewidth]{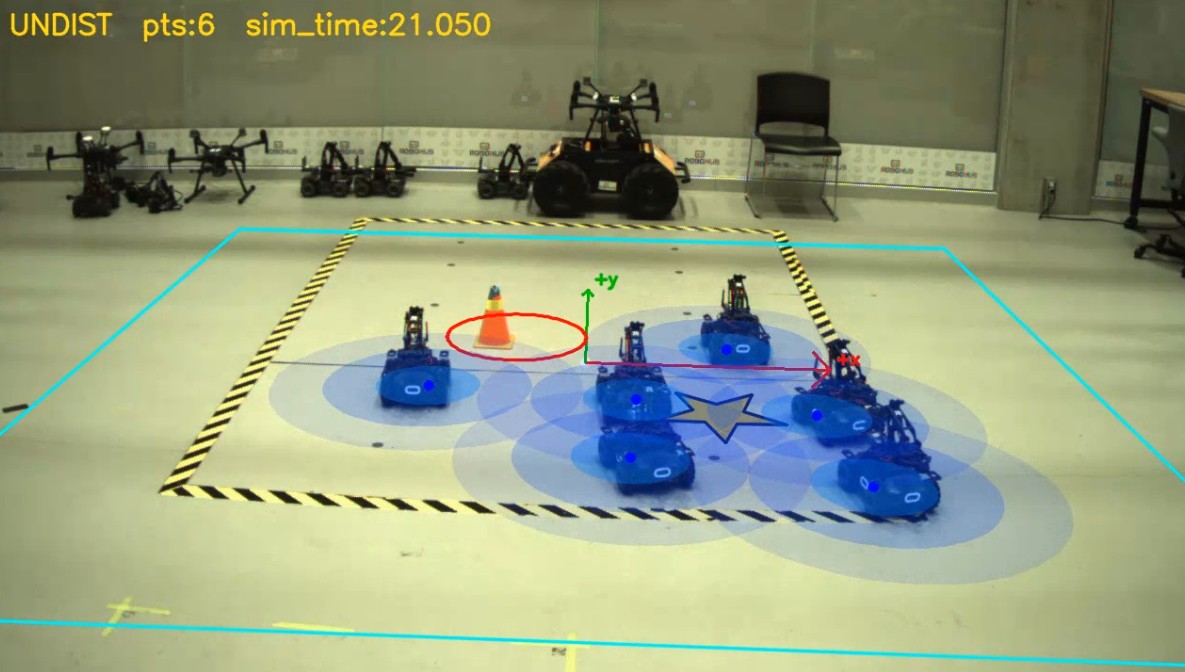}}\hspace{0.5cm}
    \subfloat{\includegraphics[height=0.2\linewidth]{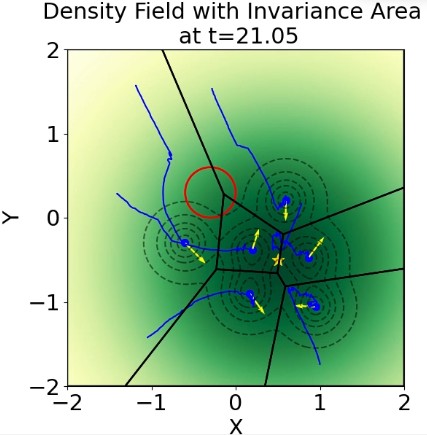}}\\
    \subfloat{\includegraphics[height=0.2\linewidth]{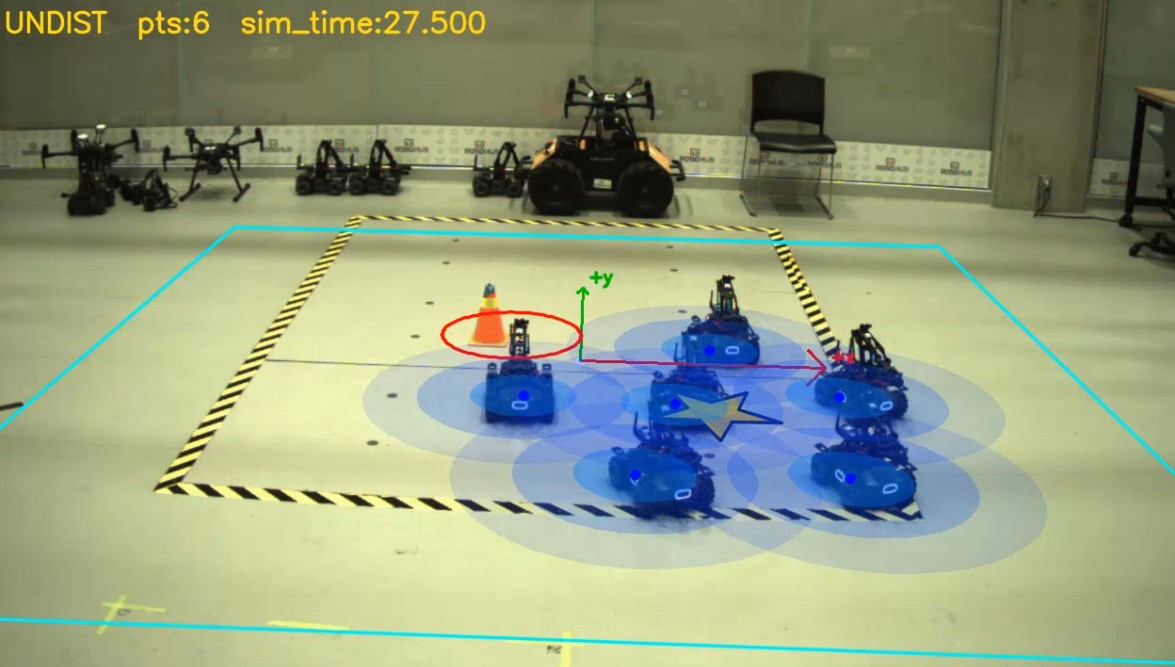}}\hspace{0.5cm}
    \subfloat{\includegraphics[height=0.2\linewidth]{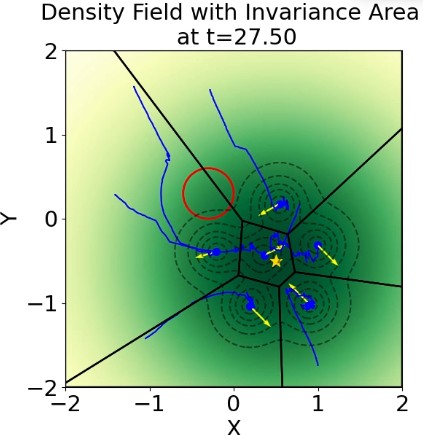}}
    \caption{Time sequence of the robot experiment. The star indicates the mean of the target, green is the target density (darker is higher), red is the invariance area, dark black line represents Voronoi cell boundaries, and blue represents robots' measured positions, with trajectory trails and black PDF contour.}
    \label{fig: exp_robot_timeSequence}
\end{figure*}

The corresponding performance evaluations are shown in Fig.~\ref{fig: performance plot experiment}. As expected, the density matching error starts high as robots are distributed around the map, and decreases as they converge towards the target PDF. The safety function decreases slightly as the three robots approach the danger area on their path to the target. As they maneuver around and past \(\mathcal{A}\), the safety measurement returns to its initial values. Notably, the controller only has access to the Voronoi cell version of measurements, \(V_c\) and \(h_c\), for computation efficiency. The post-processed true measurements, \(V\) and \(h\), have slightly lower values. The gap between the two Lyapunov measurements increases as robots are in closer proximity near the end, and the difference in safety is negligible, as stated in Remark~\ref{Remark: hc - h tiny}. 

\begin{figure}
    \centering
    \subfloat{\includegraphics[width=0.49\linewidth]{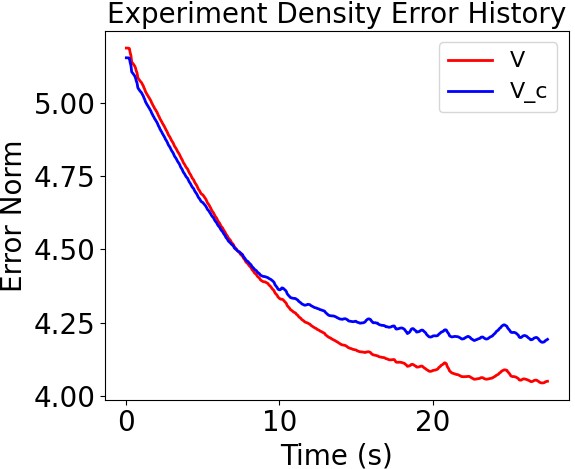}}    
    \subfloat{\includegraphics[width=0.5\linewidth]{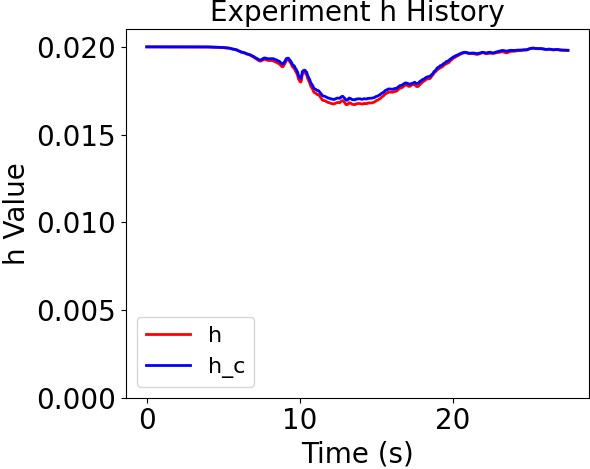}}    
    \caption{Experiment performance plots. Left: Lyapunov function \(V, V_c\), lower means matching the target better. Right: Safety barrier \(h, h_c\), higher means further from the danger area.}
    \label{fig: performance plot experiment}
\end{figure}

In this section, we demonstrated the practical applicability of our proposed controller, RC-OBC-V, with a \(6\) robot experiment. The results prove the effectiveness of the controller with real-world sensor and motion noises.

\section{Conclusion}
\label{section: Conclusion}
In this paper, we introduced a safety-critical optimization-based control strategy employing control Lyapunov and control barrier functions to guide swarm density dynamics governed by the Fokker–Planck equation. The proposed RV-OBC formulates a systematic framework that ensures accurate danger avoidance while following a predefined target, even under localization and motion noise. The RV-OBC-V variant further improves the computational efficiency and supports large-scale applications with both swarm and field size. Interestingly, our mathematical analysis and simulations reveal that the controller RV-OBC-V is more effective and safer in a sparse swarm, contrary to the common assumption that controllers work better in dense swarms.

Comparative simulations highlighted that the current state-of-the-art constrained OC method is unsuitable for safe-critical applications with sensor and motion noises. In contrast, both RV-OBC and RV-OBC-V controllers are tested in simulations with consistent success, and a \(6\) robot experiment further demonstrates the real-world applicability of RV-OBC-V. Therefore, we envision these two controllers, especially RV-OBC-V, for large-scale swarm-robotic applications such as wildfire suppression, search-and-rescue missions, and large-scale agriculture management.

\bibliography{wileyNJD-Chicago}

@INPROCEEDINGS{CBF,
  author={Ames, Aaron D. and Coogan, Samuel and Egerstedt, Magnus and Notomista, Gennaro and Sreenath, Koushil and Tabuada, Paulo},
  booktitle={2019 18th European Control Conference (ECC)}, 
  title={Control Barrier Functions: Theory and Applications}, 
  year={2019},
  volume={},
  number={},
  pages={3420-3431},
  doi={10.23919/ECC.2019.8796030}}

@article{OptimalMain,
title = {Optimal control of the Fokker-Planck equation under state constraints in the Wasserstein space},
journal = {Journal de Mathématiques Pures et Appliquées},
volume = {175},
pages = {37-75},
year = {2023},
issn = {0021-7824},
doi = {https://doi.org/10.1016/j.matpur.2023.05.002},
url = {https://www.sciencedirect.com/science/article/pii/S0021782423000594},
author = {Samuel Daudin}
}

@manual{Vicon:NexusUserGuide,
  title        = {Vicon Nexus User Guide},
  author       = {{Vicon Motion Systems Ltd.}},
  organization = {Vicon Motion Systems Ltd.},
  year         = {2024},
  version      = {Revision 6 (for Nexus 2.16)},
  url          = {https://help.vicon.com/download/attachments/11611993/Vicon%20Nexus%20User%20Guide.pdf},
  note         = {Accessed: 2025-09-26}
}

@misc{MRS_paper,
      title={Safe Decentralized Density Control of Multi-Robot Systems using PDE-Constrained Optimization with State Constraints}, 
      author={Longchen Niu and Gennaro Notomista},
      year={2025},
      eprint={2510.20643},
      archivePrefix={arXiv},
      primaryClass={eess.SY},
      url={https://arxiv.org/abs/2510.20643}, 
      note= {Presented at the IEEE International Symposium on Multi-Robot and Multi-Agent Systems (MRS), 2025}

}

@manual{DJI:RoboMasterS1UserManual,
  title        = {RoboMaster S1 User Manual},
  author       = {{DJI}},
  organization = {DJI},
  year         = {2020},
  version      = {1.8},
  url          = {https://www.dji.com/ca/downloads/products/robomaster-s1#doc},
  note         = {Accessed: 2025-09-26}
}

@manual{mosek,
   author  = "{MOSEK ApS}",
   title = {MOSEK Fusion API for Python 10.2.16},
   year = {2024},
   url = {https://docs.mosek.com/10.2/pythonfusion/index.html}
 }

@article{Fire,
  title={The use of swarms of unmanned aerial vehicles in mitigating area coverage challenges of forest-fire-extinguishing activities: A systematic literature review},
  author={Alsammak, Ihab L Hussein and Mahmoud, Moamin A and Aris, Hazleen and AlKilabi, Muhanad and Mahdi, Mohammed Najah},
  journal={Forests},
  volume={13},
  number={5},
  pages={811},
  year={2022},
  publisher={MDPI}
}

@article{briani2016stablesolutionspotentialmean,
  title={Stable solutions in potential mean field game systems},
  author={Briani, Ariela and Cardaliaguet, Pierre},
  journal={Nonlinear Differential Equations and Applications NoDEA},
  volume={25},
  number={1},
  pages={1},
  year={2018},
  publisher={Springer}
}

@ARTICLE{CLF_PDE_NoDensity,
  author={Yamaguchi, Kaiyo and Endo, Takahiro and Matsuno, Fumitoshi},
  journal={IEEE Transactions on Control Systems Technology}, 
  title={Formation Control of Multiagent System Based on Higher Order Partial Differential Equations}, 
  year={2022},
  volume={30},
  number={2},
  pages={570-582},
  doi={10.1109/TCST.2021.3068401}}

@article{safety_car_distance,
  title={Control input design for a robot swarm maintaining safety distances in crowded environment},
  author={Origane, Yuki and Hattori, Yuya and Kurabayashi, Daisuke},
  journal={Symmetry},
  volume={13},
  number={3},
  pages={478},
  year={2021},
  publisher={MDPI}
}

@article{wang2017safety,
  title={Safety barrier certificates for collisions-free multirobot systems},
  author={Wang, Li and Ames, Aaron D and Egerstedt, Magnus},
  journal={IEEE Transactions on Robotics},
  volume={33},
  number={3},
  pages={661--674},
  year={2017},
  publisher={IEEE}
}

@INPROCEEDINGS{CBF-PDE2ODE,
  author={Park, Younghwa and Sloth, Christoffer},
  booktitle={2023 11th International Conference on Control, Mechatronics and Automation (ICCMA)}, 
  title={Discretization-Robust Safety Barrier of Partial Differential Equation}, 
  year={2023},
  volume={},
  number={},
  pages={49-54},
  doi={10.1109/ICCMA59762.2023.10374980}}

@article{Optimal_withInvariance,
title = {Robust optimal density control of robotic swarms},
journal = {Automatica},
volume = {176},
pages = {112218},
year = {2025},
issn = {0005-1098},
doi = {https://doi.org/10.1016/j.automatica.2025.112218},
author = {Carlo Sinigaglia and Andrea Manzoni and Francesco Braghin and Spring Berman}
}

@INPROCEEDINGS{Optimal_noInvariance,
  author={Elamvazhuthi, Karthik and Berman, Spring},
  booktitle={2015 IEEE International Conference on Robotics and Automation (ICRA)}, 
  title={Optimal control of stochastic coverage strategies for robotic swarms}, 
  year={2015},
  volume={},
  number={},
  pages={1822-1829},
  doi={10.1109/ICRA.2015.7139435}}

@ARTICLE{RalPaper,
 author={Niu, Longchen and Notomista, Gennaro},
  journal={IEEE Robotics and Automation Letters}, 
  title={Decentralized Density Control of Multi-Robot Systems Using PDE-Constrained Optimization}, 
  year={2025},
  volume={10},
  number={4},
  pages={4045-4052},
  doi={10.1109/LRA.2025.3548501}}

@article{exisitenceMeanFieldGameWithSTate,
author="Cannarsa, Piermarco
and Capuani, Rossana",
editor="Cardaliaguet, Pierre
and Porretta, Alessio
and Salvarani, Francesco",
title="Existence and Uniqueness for Mean Field Games with State Constraints",
bookTitle="PDE Models for Multi-Agent Phenomena",
year="2018",
publisher="Springer International Publishing",
address="Cham",
pages="49--71",
isbn="978-3-030-01947-1",
doi="10.1007/978-3-030-01947-1_3",
url="https://doi.org/10.1007/978-3-030-01947-1_3"
}

@article{quasilinear_parabolic_PDE_OptimalControl_StateConstraint,
author = {Hoppe, Fabian and Neitzel, Ira},
title = {Optimal Control of Quasilinear Parabolic PDEs with State-Constraints},
journal = {SIAM Journal on Control and Optimization},
volume = {60},
number = {1},
pages = {330-354},
year = {2022},
doi = {10.1137/20M1383951},
URL = { 
        https://doi.org/10.1137/20M1383951
},
eprint = { 
        https://doi.org/10.1137/20M1383951
}
}

@INPROCEEDINGS{8206299,
  author={Li, Hanjun and Feng, Chunhan and Ehrhard, Henry and Shen, Yijun and Cobos, Bernardo and Zhang, Fangbo and Elamvazhuthi, Karthik and Berman, Spring and Haberland, Matt and Bertozzi, Andrea L.},
  booktitle={2017 IEEE/RSJ International Conference on Intelligent Robots and Systems (IROS)}, 
  title={Decentralized stochastic control of robotic swarm density: Theory, simulation, and experiment}, 
  year={2017},
  volume={},
  number={},
  pages={4341-4347},
  doi={10.1109/IROS.2017.8206299}}

@article{doi:10.1177/0278364908100177,
author = {Mac Schwager and Daniela Rus and Jean-Jacques Slotine},
title ={Decentralized, Adaptive Coverage Control for Networked Robots},
journal = {The International Journal of Robotics Research},
volume = {28},
number = {3},
pages = {357-375},
year = {2009},
doi = {10.1177/0278364908100177},
eprint = {https://doi.org/10.1177/0278364908100177}}

@article{TERUEL201951,
title = {A distributed robot swarm control for dynamic region coverage},
journal = {Robotics and Autonomous Systems},
volume = {119},
pages = {51-63},
year = {2019},
issn = {0921-8890},
doi = {https://doi.org/10.1016/j.robot.2019.06.002},
author = {Enrique Teruel and Rosario Aragues and Gonzalo López-Nicolás}}

@article{Elamvazhuthi_Berman_2019, 
title={Mean-field models in Swarm Robotics: A survey}, 
volume={15}, 
DOI={10.1088/1748-3190/ab49a4}, 
number={1}, 
journal={Bioinspiration \& Biomimetics}, 
author={Elamvazhuthi, Karthik and Berman, Spring}, 
year={2019}, 
month={Nov}, 
pages={015001}}

@article{Archer2004DynamicalDF,
  title={Dynamical density functional theory for interacting Brownian particles: stochastic or deterministic?},
  author={Andrew J. Archer and Markus Rauscher},
  journal={Journal of Physics A},
  year={2004},
  volume={37},
  pages={9325-9333}
}

@article{ANNUNZIATO2013487,
title = {A Fokker–Planck control framework for multidimensional stochastic processes},
journal = {Journal of Computational and Applied Mathematics},
volume = {237},
number = {1},
pages = {487-507},
year = {2013},
issn = {0377-0427},
doi = {https://doi.org/10.1016/j.cam.2012.06.019},
author = {M. Annunziato and A. Borzì}}

@INPROCEEDINGS{application_enivronment,
  author={Carpentiero, Marco and Gugliermetti, Luca and Sabatini, Marco and Palmerini, Giovanni B.},
  booktitle={2017 IEEE 14th International Conference on Networking, Sensing and Control (ICNSC)}, 
  title={A swarm of wheeled and aerial robots for environmental monitoring}, 
  year={2017},
  volume={},
  number={},
  pages={90-95},
  doi={10.1109/ICNSC.2017.8000073}}

@article{application_SaR,
title = {Behavior-based swarm robotic search and rescue using fuzzy controller},
journal = {Computers \& Electrical Engineering},
volume = {70},
pages = {53-65},
year = {2018},
issn = {0045-7906},
doi = {https://doi.org/10.1016/j.compeleceng.2018.06.003},
author = {Ahmad Din and Meh Jabeen and Kashif Zia and Abbas Khalid and Dinesh Kumar Saini}}

@INPROCEEDINGS{application_agriculture,
  author={Albani, Dario and IJsselmuiden, Joris and Haken, Ramon and Trianni, Vito},
  booktitle={2017 14th IEEE International Conference on Advanced Video and Signal Based Surveillance (AVSS)}, 
  title={Monitoring and mapping with robot swarms for agricultural applications}, 
  year={2017},
  volume={},
  number={},
  pages={1-6},
  doi={10.1109/AVSS.2017.8078478}}

@inproceedings{application_med,
  author    = {Mohammad Majid al-Rifaie and
               Ahmed Aber and
               Remigijus Raisys},
  title     = {Swarming Robots and Possible Medical Applications},
  booktitle = {Proceedings of the 17th International Symposium on Electronic Art ({ISEA} 2011)},
  year      = {2011},
  address   = {Istanbul, Turkey}
}

@article{reis2020control,
  title={Control barrier function-based quadratic programs introduce undesirable asymptotically stable equilibria},
  author={Reis, Matheus F and Aguiar, A Pedro and Tabuada, Paulo},
  journal={IEEE Control Systems Letters},
  volume={5},
  number={2},
  pages={731--736},
  year={2020},
  publisher={IEEE}
}
\end{document}